\documentclass[12pt, draftclsnofoot, onecolumn]{IEEEtran}
\usepackage[colorlinks,
linkcolor=red,
anchorcolor=green,
citecolor=blue
]{hyperref} 
\usepackage{amsmath,amssymb}
\usepackage{latexsym}
\usepackage{CJK}
\usepackage{cite}

\usepackage{color, soul}

\usepackage{graphicx}
\usepackage{epstopdf}
\usepackage{epsfig}
\usepackage{subfigure} 

\usepackage{verbatim}  

\usepackage{mathrsfs}	
\usepackage{algorithm} 
\usepackage{algorithmic} 
\usepackage{booktabs}
\usepackage{textcomp}
\usepackage{multirow}
\usepackage{lettrine}

\usepackage{mathrsfs}
\usepackage{amsmath}
\usepackage{amssymb}

\usepackage{algorithm}
\usepackage{algorithmic}

\usepackage{amsthm} 

\DeclareMathOperator*{\argmine}{arg\, }

\newtheorem{thm}{Theorem}
\newtheorem{cor}{Corollary}
\newtheorem{lem}{Lemma}
\newtheorem{prop}{Proposition}

\begin{document}
\title{Toward Traffic Patterns in High-speed Railway Communication Systems: Power Allocation and Antenna Selection}



\author{Jiaxun~Lu, Ke~Xiong,~\IEEEmembership{Member, ~IEEE}, Xuhong~Chen, Pingyi~Fan,~\IEEEmembership{Senior Member, ~IEEE}\\

\thanks{
Jiaxun~Lu, Ke~Xiong, Xuhong~Chen and Pingyi Fan are with the Department of Electronic Engineering,  Tsinghua University,  Beijing,  R.P. China, 100084. e-mail: lujx14@mails.tsinghua.edu.cn, kxiong@bjtu.edu.cn, chenxh13@mails.tsinghua.edu.cn, fpy@tsinghua.edu.cn. Ke~Xiong is also with the School of Computer and Information Technology,  Beijing Jiaotong University,  Beijing,  R.P. China, 100084.
}
}

\maketitle

\graphicspath{{Figures/}}

\begin{abstract}

In high-speed railway (HSR) communication systems, distributed antenna is usually employed to support frequent handover and enhance the signal to noise ratio to user equipments. In this case, dynamic time-domain power allocation and antenna selection (PAWAS) could be jointly optimized to improve the system performances. This paper consider this problem in such a simple way where dynamic switching between multiple-input-multiple-output (MIMO) and single-input-multiple-output (SIMO) is allowed and exclusively utilized, while the channel states and traffic demand are taken into account. The channel states includes sparse and rich scattering terrains, and the traffic patterns includes delay-sensitive and delay-insensitive as well as hybrid. Some important results are obtained in theory. In sparse scattering terrains, for delay-sensitive traffic, the PAWAS can be viewed as the generalization of channel-inversion associated with transmit antenna selection. On the contrary, for delay-insensitive traffic, the power allocation with MIMO can be viewed as channel-inversion, but with SIMO, it is traditional water-filling. For the hybrid traffic, the PAWAS can be partitioned as delay-sensitive and delay-insensitive parts by some specific strategies. In rich scattering terrains, the corresponding PAWAS is derived by some amendments in sparse scattering terrains and similar results are then presented.

\end{abstract}

\begin{IEEEkeywords}
Mobility, green communication, queening theory, traffic demand, power allocation, antenna selection.
\end{IEEEkeywords}

\IEEEpeerreviewmaketitle

\section{Introduction}\label{Sec:Introduction}

\lettrine[lines=2]{M}{obility} and green are two important indexes in 5G\cite{andrews2014will}. As one of the typical high mobility scenario, high-speed railway (HSR) is experiencing explosive growth in recent decades, where distributed antenna system (DAS) is usually employed to avoid frequent handover and improve signal to noise ratio (SNR) for receivers \cite{yeh2010theory,wang2012distributed}. In HSR system, the high velocity of train and large spacing between adjacent antennas, leading to the fast time-varying and wide fluctuant large-scale path fading, create special challenges for efficient energy utilization, i.e. green communication. A general strategy to combat these detrimental effects is adopting dynamic resource allocation based on the channel states, e.g. transmit power allocation and antenna selection. For example, time-domain water-filling was adopted to achieve the best energy-efficiency in time-varying or frequency-selective channels and antenna selection could be used to adaptively find the best antenna pattern to transmit with.

In the literature, most existing works on dynamic resource allocation separately considered the power allocation and antenna selection for HSR. In this paper, we address the problem on how to minimize average transmit power by jointly optimizing time-domain power allocation associated with antenna selection (PAWAS) in high mobility cases? We focus on the downlink scenario, where multiple mobile relays (MRs) are mounted on the carriages, forming the two-hop architecture and multiple-input channel\cite{zhang2015optimal,lu2016location}. In this case, when single transmit antenna is selected, it's a single-input-multiple-output (SIMO) transmission, whereas when multiple transmit antennas are selected, it's a multiple-input-multiple-output (MIMO) transmission. In discussed case, it is allowed to adaptively switch the transmit mode between SIMO and MIMO and optimally allocate power along time and among antennas according to system traffic demand, where the different traffic patterns will be considered.

The first one is the delay-sensitive, where the delay limit is small compared to the time-scale of channel fading. In this category, the arrived data needs instant service to meet its delay constraint. The second one is delay-insensitive, where delay is not an key issue and with relatively loose requirement. For instance, in the recently proposed caching technique\cite{shanmugam2013femtocaching,gitzenis2013asymptotic,cui2016analysis}, where popular contents can be cached in the MRs to support user's quick access to massive content, the delay of caching is insignificant. More generally, we characterize the traffic with delay limit longer than one operation duration of HSR as delay-insensitive traffic.

In delay-sensitive category, we prove that the time-domain power allocation can be viewed as channel-inversion associate with transmit antenna selection, since the transmit capacity in each time-slot is determined by relevant traffic arrival rate and queening delay. By contrast, in delay-insensitive category, one difficulty is that the energy efficiency of MIMO and SIMO varies with respect to allocated power and position of train, and hence it's hard for us to select the optimal transmit antennas corresponding to minimized average power, especially when the exact allocated power at each time-slot is not prior known. In the analyzed HSR system, we show that despite the diversity gain of $2\times2$ MIMO system doubles that in $1\times2$ SIMO systems\footnote{Consider the long distance between adjacent transmit antennas, at most two transmit antennas can be achieved by MRs. For practical deployment, two MRs are equipped on the carriage.}, the received signal power of MIMO can be rewritten as the product of transmit power and \emph{effective channel gain}, which depends only on the transmit power and fading coefficients of MIMO channel. This is first observed in this paper.

Applying the notion of effective channel gain, we derive the optimal PAWAS for delay-insensitive traffic. The transmit antenna is selected by simply comparing the effective channel gain of MIMO with the actual channel gain of SIMO. We show that when MIMO is selected, the optimal power allocation can be viewed as the generalization of channel-inversion, whereas when SIMO is selected, it is traditional water-filling. Practical traffic usually include both delay-sensitive and -insensitive traffic simultaneously. To this hybrid traffic pattern, we prove that the optimal PAWAS can be partitioned into delay-sensitive and -insensitive parts by some specific strategies.

Nevertheless, above results are derived in sparse scattering scenarios, such as viaducts and wide plains. This paper also considers rich scattering scenarios, such as mountain, tunnels, urban and sub-urban districts, where the small-scale propagation channel can be modeled as Nakagami-$m$ channel\cite{li2013channel,dong2012varepsilon,liu2014novel}. We prove that the ergodic capacity loss is related to specific antenna schemes and $m$: the capacity loss in $2\times2$ MIMO doubles that in $1\times2$ SIMO with same $m$. By few amendments, previous results can be easily extended to rich scattering scenarios. Moreover, simulation results show that MIMO is suggested when train is near the middle of adjacent transmit antennas or in cases with heavy traffic demand. This is because, more degree of freedom in MIMO reduces instantaneous transmit power. In both sparse and rich scattering scenarios, we will show that our proposed PAWAS can provide minimized average transmit power for arbitrary traffic demand. In addition, we also consider the maximum transmit power constraint and prove that low value of maximum transmit power may greatly reduce the effectiveness of PAWAS.

In order to highlight our contributions compared with existing works, we stress some existing works as follows. Dong \emph{et al.} introduced time-domain power allocation to HSR in \cite{Dong2014Efficiency} with respect to proportionally fair criterion and Zhang \emph{et al.} reconsidered it for minimizing average transmit power with delay constraints\cite{zhang2015optimal}. The optimal power allocation is derived in single-input-single-output (SISO) case. Also, \cite{zhang2015optimal} assumed deterministic arrived traffic with hard delay constraint. Our results are with DAS and multiple receive antennas, and two-layer traffic mode is involved with more practical stochastic arrived data and average delay constraint. In \cite{li2015qos}, the authors proposed the power allocation in SISO for two different traffic patterns discussed in this paper, but only the transmission rate is considered. This paper matches the arrival and transmission process by queening theory and derive the optimal service mode for arbitrary arrived data. In \cite{heath2001antenna,sanayei2004antenna}, the authors analyzed the performance of MIMO and SIMO, and the antenna selection strategies were proposed to maximize channel capacity. In \cite{hui2014efficiency}, the antenna selection method in static large-scale MIMO scenario was proposed for minimizing average transmit power. Since this paper considers antenna selection as train running along the railway, it can be viewed as a generalization in dynamic scenario. Specifically, when the velocity of train is close to zero, our proposed method can also be applied to previous static scenario.

The rest of this paper is organized as follows. { In Section \ref{Sec:SystemModel}, the system model is introduced. Also, a simple example is shown to illustrate the necessity of PAWAS, and the problem on optimal PAWAS is formulated. In section \ref{Sec:PAWASInDSAWGN}, PAWAS for delay-sensitive, -insensitive and hybrid traffic in sparse scattering scenarios is proposed, respectively. Section \ref{Sec:PAWASInHybridNakagami} considers rich scattering scenarios under Nakagami-$m$ fading model and proposes relevant PAWAS. In Section \ref{Sec:NumericalRes}, the validity of previous theoretical results and the effectiveness of our derived PAWAS are verified by numerical results. Finally, conclusions are given in Section \ref{Sec:Conclusion}.}

\emph{Notations}: $(\cdot)^{\dagger}$ and $(\cdot)^{+}$ denote the conjugate transpose of $(\cdot)$ and $\max((\cdot),0)$, respectively. The symbols $\mathbf{I}$ denotes the identical matrix. $\det(\cdot)$ is the determinant of $(\cdot)$. $(\cdot)\equiv(\ast)$ means $(\cdot)$ is equivalent to $(\ast)$. $(\cdot)\doteq(\ast)$ means $(\cdot)$ equals to $(\ast)$ by definition. $\mathbb{E}(\cdot)$ is the expectation of $(\cdot)$.

\section{System Model and Problem Formulation}\label{Sec:SystemModel}

{ This section first illustrates the network model of downlink in HSR communication, where the DAS and two-hop architecture are adopted. Then, the system capacity relevant to MIMO and SIMO are presented, and the problem on optimal PAWAS is stressed and mathematically formulated.}

\subsection{Network Model}\label{Sec:NetWorModel}

\begin{figure*}[htbp]
\centering
\includegraphics[width=0.85\textwidth]{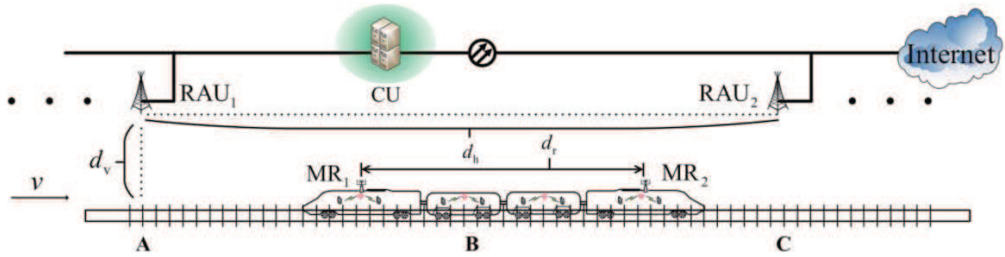}
\caption{Distributed antennas on railway coverage system, where the RAUs are connected to CU with optical fibers.} \label{Fig:DistributedAntenna}
\end{figure*}

The downlink of DAS is shown in Fig. \ref{Fig:DistributedAntenna}, where the central unit (CU) is connected with several remote antenna units (RAUs) by optical fibers\cite{lu2016location,yeh2010theory,wang2012distributed}. At the transmit side, CU is connected with Internet and sends the base band signals to corresponding RAUs, which modulate and transmit signals to receivers. As shown in Fig. \ref{Fig:DistributedAntenna}, each RAU is located at a vertical distance of $d_v$ from the railway and the distance between adjacent RAUs is $d_h$. { Due to the limited per antenna module power\cite{choi2007downlink}, the maximum normalized transmit power of each RAU is constrained by $\mathcal{P}_{\rm{max}}$.} In this paper, if not specified, the transmit power is normalized by double-side noise power.

\begin{figure}[htbp]
\centering
\includegraphics[width=0.4\textwidth]{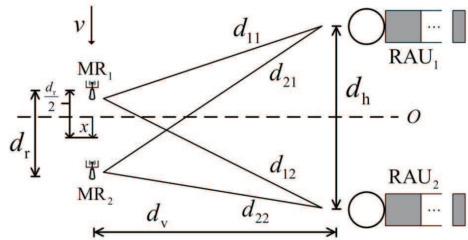}
\caption{The positions of each node on a reference coordinate system.} \label{Fig:MultipleAntennaModel}
\end{figure}

At the receive side, { the popular two-hop architecture\cite{wang2012distributed,zhang2015optimal} is adopted to enhance communication quality.} In the downlink, RAUs first transmit signals to the mobile relays (MRs) deployed on the carriage. Then MRs forward signals to user equipments inside train. As depicted in Fig. \ref{Fig:MultipleAntennaModel}, the distance between them is $d_{\rm{r}}$, and the distance between ${\rm{MR}}_i$ and ${\rm{RAU}}_j$ is denoted as $d_{ij}$. Let the dash-line located at the middle of two RAUs as the original ($O$), the vertical distance between the middle of two MRs and $O$ is $x=vt$. In this paper, because the distance between MRs and RAUs varies periodically and is symmetric about $O$, we take the half period from 0 to ${T}/{2}$ for example, where $T=d_h/v$.

\subsection{Capacity Model}\label{SubSec:CapacityModelOfAWGN}

{ As shown in Fig. \ref{Fig:MultipleAntennaModel}, when both $\rm{RAU}_1$ and $\rm{RAU}_2$ are selected to transmit signals, it's MIMO transmission. When only $\rm{RAU}_2$ is selected to transmit signals, it's SIMO transmission.} Define the channel fading matrix of MIMO as
\begin{equation}\label{equ:HExpansion}
\mathbf{H}(t) = \begin{bmatrix} h_{11}(t) & h_{12}(t) \\ h_{21}(t) & h_{22}(t) \end{bmatrix},
\end{equation}
where $h_{ij}$ ($i,j=1, 2$) is the channel coefficient between the $j$-th $\rm{RAU}$ and the $i$-th MR. When $\rm{RAU}_1$ and $\rm{RAU}_2$ transmit signals independently with equal power, the system capacity can be expressed as\cite{lu2014precoding,lu2016subcarrier}
\begin{equation}\label{equ:MIMOCapacityExpressionOriginal}
C_{\rm{M}}(t) = \log_2 \det \left( \mathbf{I} + \begin{bmatrix} \mathcal{P}(t)/2 & 0 \\ 0 & \mathcal{P}(t)/2 \end{bmatrix}  \mathbf{H}(t) ^{\dagger} \mathbf{H}(t)\right),
\end{equation}
where $\mathcal{P}(t)$ is the overall transmit power. Substitute \eqref{equ:HExpansion} into \eqref{equ:MIMOCapacityExpressionOriginal}, the MIMO capacity can be rewritten as
\begin{equation}\label{equ:MIMOCapacityLOS}
C_{\rm{M}}(t) = \log_2 \Big\{ \frac{\alpha_1\alpha_2 - \beta^2}{4} \mathcal{P}(t)^2 + \frac{\alpha_1 + \alpha_2}{2} \mathcal{P}(t) + 1 \Big\},
\end{equation}
where
\begin{equation}\label{equ:HHExpansionElements}
\begin{cases}
\alpha_1 &= h_{11}^2 + h_{21}^2, ~ \alpha_2 = h_{12}^2 + h_{22}^2,\\
\beta &= h_{11}h_{12} + h_{21}h_{22}.\\
\end{cases}
\end{equation}
Note that the index $t$ of $h_{ij}(t)$ and $\alpha_{i}(t)$ ($i,j=1,2$) is omitted for the concise of paper.

In SIMO with maximal ratio combining, the maximum achievable capacity can be expressed by
\begin{equation}\label{equ:SIMOCapacityLOS}
C_{\rm{S}}(t) = \log_2\big( 1 + \mathcal{P}(t) \alpha_2 \big),
\end{equation}
where $\alpha_2$ is defined in \eqref{equ:HHExpansionElements}.


\subsection{Problem Formulation}
{ In this part, we shall first stress the necessity of joint power allocation and antenna selection. Then, the PAWAS problem is mathematically formulated.}

\begin{figure}[tbp]
\centering
\subfigure[]{ \label{Fig:DS9bpcu}
\includegraphics[width=0.3\columnwidth]{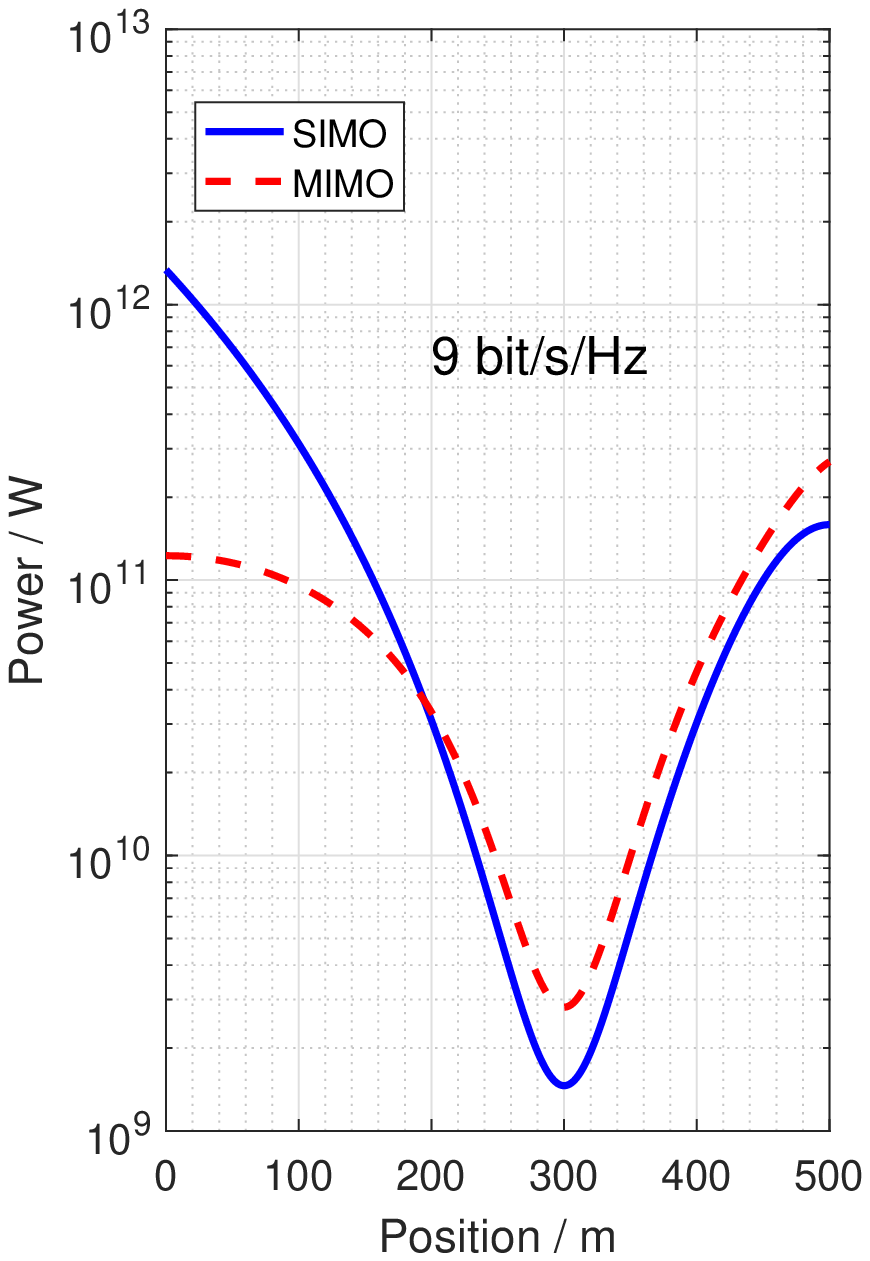}}
\subfigure[]{ \label{Fig:DS12bpcu}
\includegraphics[width=0.3\columnwidth]{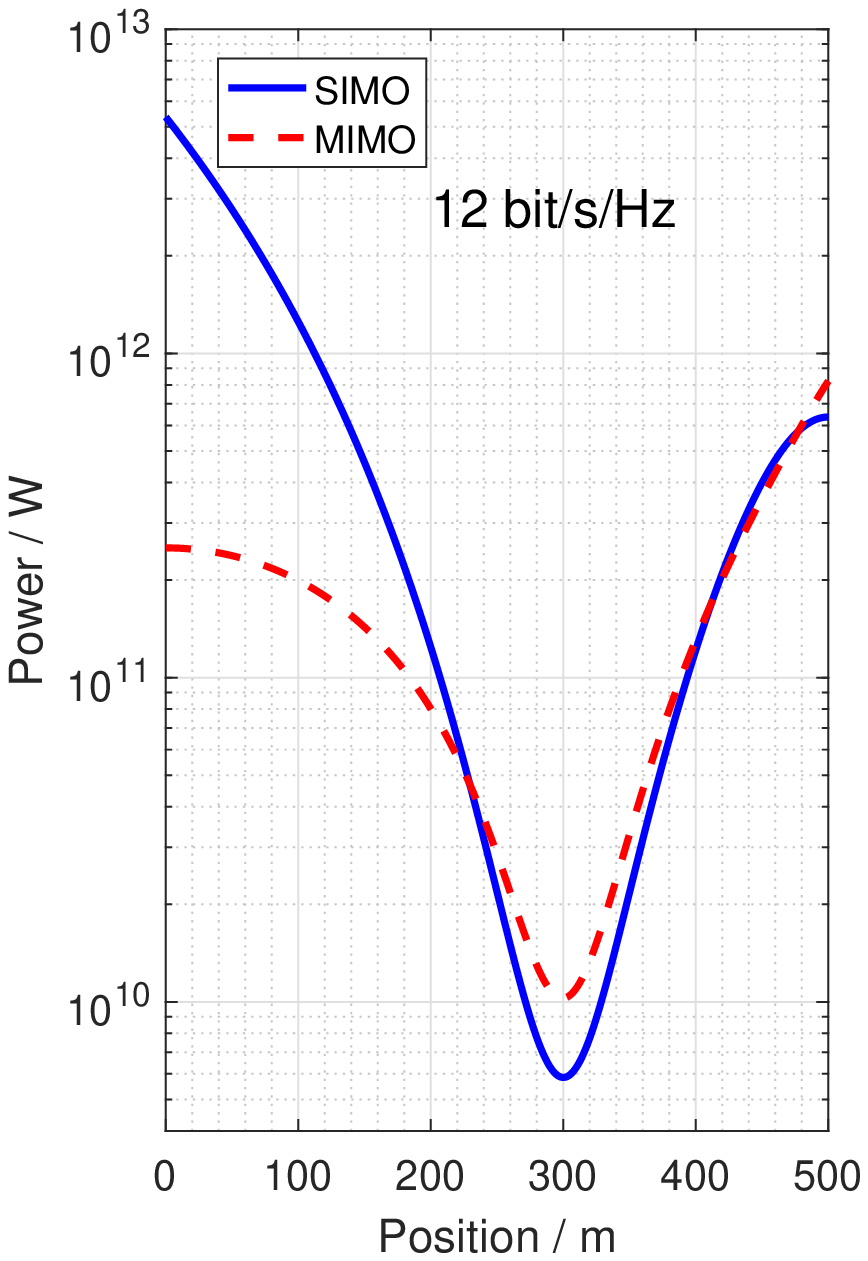}}
\caption{The transmit power of SIMO and MIMO schemes when providing constant service rates. In (a) and (b), C(t) is 9 and 12 bit/s/Hz, respectively. In simulations, $d_{\rm{r}}$=400 m, $d_h$=1000 m and $d_v$=100 m. The fading model is shown in \eqref{equ:LOS_h_ij} and $\iota$=3.8.} \label{Fig:DSInVariousCapacity}
\end{figure}
\subsubsection{A Simple Example}~~

Fig. \ref{Fig:DSInVariousCapacity} shows the transmit power of MIMO and SIMO in sparse scattering scenarios when providing { service rate with 9 bit/s/Hz and 12 bit/s/Hz, respectively.} The dashed lines demonstrate the transmit power of MIMO, and the solid lines are the power of SIMO. { In Fig. \ref{Fig:DS9bpcu}, one can see that when $x\in[0,200]$, MIMO consumes less power than SIMO, because in this case MRs are near to both transmit antennas, and hence MIMO provides more diversity gain to enhance power efficiency. However, in the rest cases, where $\rm{RAU}_1$ is relatively far away from MRs, SIMO may outperform MIMO. In Fig. \ref{Fig:DS12bpcu}, it shows that with increased service rate, MIMO consumes less power in more cases.}

{ With this example, we can conclude that the optimal antenna scheme may vary with respect to the location of train and the service rate demand. More specifically, in cases with less distance to RAUs or with higher service rates, MIMO consumes less power than SIMO.}\footnote{The quantized relations will be analyzed in following sections.} Nevertheless, this example is based on constant service rate. When time-domain power allocation is involved, the system service rate may vary with positions according to its overall service amount demand. That is, the optimal PAWAS is supposed to be \emph{location} and \emph{service demand} aware. Also, to meet the service demand with minimized average transmit power in half period (i.e. $[0,T/2]$), time-domain power allocation and antenna selection need to be jointly optimized.

\subsubsection{Mathematical Formulation}\label{Sec:subsubMathmaticalFormulation}~~

{ The Poisson distributed traffic arrivals and negative exponential distributed package length are assumed.} As Berry's assumption in \cite{berry2013optimal}, the delay of delay-sensitive traffic is small relative to the coherence time of large-scale fading. For delay-sensitive traffic, the arrival process in Fig. \ref{Fig:MultipleAntennaModel} is modeled as M/M/1 queuing system. Let $\tau_{\rm{max}}$ be the maximum average queening delay, and denote the arrival rate and service rate of delay-sensitive traffic as $\lambda_{\rm{s}}$ and $\mu_{\rm{s}}$, respectively. The relevant traffic load can be expressed as $\rho_{\rm{s}}=\lambda_{\rm{s}}/\mu_{\rm{s}}$ and the average queuing delay of delay-sensitive traffic can be expressed as\cite{kleinrock1975queueing}
\begin{equation}\label{equ:QueDelay}
\tau = \frac{1}{\mu_{\rm{s}}(1-\rho_{\rm{s}})}\leq \tau_{\rm{max}}.
\end{equation}

The arrival rate and service rate of delay-insensitive traffic are represented by $\lambda_{\rm{i}}$ and $\mu_i(t)$, respectively. To meet the instantaneous traffic demand, the system capacity satisfies $C(t)=\bar{L}(\mu_{\rm{s}} + \mu_i(t))$, where $\bar{L}$ is the average package length normalized by transmit bandwidth. Notice that to minimize average transmit power, $\mu_i(t)$ should be adaptively adjusted with respect to the varying path loss as train running along railway. Moreover, to avoid infinite queuing delay of delay-insensitive traffic, we let $\int_0^{T/2}\mu_i(t) \geq \lambda_{\rm{i}} T/2$. Since for delay-sensitive traffic, there exists $\int_0^{T/2}\mu_s = \lambda_s T/2$, the third constraint shown later in \eqref{equ:MainOptimizeProblemConstraints} can be easily derived.

To demonstrate the optimal PAWAS, we give an illustration of antenna selection first. Observing \eqref{equ:MIMOCapacityLOS} and \eqref{equ:SIMOCapacityLOS}, we can use a general form to express the system capacity for the considered train-ground channel, i.e.
\begin{equation}\label{equ:UniversalCapacity}
C(t) = \log_2\left( 1 + \Gamma(\mathcal{P}(t),t) \mathcal{P}(t) \right),
\end{equation}
where $\Gamma(\mathcal{P}(t),t)$ is the \emph{effective channel gain} defined by
\begin{equation}\label{equ:equivalentChannelFading}
\Gamma(\mathcal{P}(t),t) =  \max\left(\Big(\frac{\alpha_1\alpha_2 - \beta^2}{4}{\mathcal{P}(t)} + \frac{\alpha_1 + \alpha_2}{2}\Big) , \alpha_2\right).
\end{equation}
The first part can be viewed as the generalized channel gain of MIMO, and the second part is the channel gain of SIMO. Define $\Theta$ as the antenna selection mode. When the first part is greater than the second part, system select MIMO and $\Theta=0$. Otherwise, system select SIMO and $\Theta=1$. The antenna selection scheme can be expressed as follows.

\begin{lem}\label{lem:PThreshold_DelayConstrainedMDSInAWGN}
When $\mathcal{P}(t) \geq \zeta_{\mathcal{P}}$, MIMO should be selected, where $\zeta_{\mathcal{P}}$ is the power threshold defined as
\begin{equation}\label{equ:MDSPThreshold}
\zeta_{\mathcal{P}} = \frac{2(\alpha_2-\alpha_1)}{\alpha_1\alpha_2-\beta^2}.
\end{equation}
An equivalent antenna selection criterion is that when $C(t) \geq \zeta_c$, system select MIMO. $\zeta_c$ is the capacity threshold defined as
\begin{equation}\label{equ:MDSRThreshold}
\zeta_c = \log_2\Big( 1+ \frac{2\alpha_2(\alpha_2-\alpha_1)}{\alpha_1\alpha_2-\beta^2}\Big).
\end{equation}
The expression of $\alpha_1$, $\alpha_2$ and $\beta$ are shown in \eqref{equ:HHExpansionElements}, \eqref{equ:LOS_h_ij} and \eqref{equ:distancec_ij}, respectively. More over, when $\mathcal{P}(t) = \zeta_{\mathcal{P}}$, $C(t) = \zeta_c$.
\end{lem}

\begin{proof}
From \eqref{equ:equivalentChannelFading}, Eqn. \eqref{equ:MDSPThreshold} can be easily derived. Substitute \eqref{equ:MDSPThreshold} into \eqref{equ:SIMOCapacityLOS}, we can derive Eqn. \eqref{equ:MDSRThreshold}. This completes the proof.
\end{proof}

It can be observed that system always adopt antenna scheme with higher channel gain. That is, the selected antenna scheme can provide higher transmit capacity. In other words, when providing same transmit capacity, the selected antenna scheme minimizes instantaneous transmit power. Also, it can be seen from Lemma \ref{lem:PThreshold_DelayConstrainedMDSInAWGN} that the antenna selection scheme in one time-slot is coupled with instantaneous transmit power. That is, the antenna selection and transmit power need to be jointly optimized, which is expressed by
\begin{flalign}\label{equ:MainOptimizeProblem}
\mathbf{P}_1: \,\,\min_{\Theta,\mathcal{P}(t)} \,\,\, &\frac{2}{T} \int_{0}^{\frac{T}{2}}\mathcal{P}(t) dt\\
\textrm{s.t.}\,\,\,&\eqref{equ:QueDelay},\eqref{equ:UniversalCapacity},C(t)=\bar{L}(\mu_s+\mu_i(t)),\nonumber\\
& \mathcal{P}(t)  \leq \mathcal{P}_{\rm{max}}(t),\nonumber\\
&\int_{0}^{T/2} C(t)  \geq \bar{L}(\lambda_{\rm{s}} + \lambda_{\rm{i}})T/2.\label{equ:MainOptimizeProblemConstraints}
\end{flalign}
The first constraint is the average queuing delay limitation, and the second one is the maximum instantaneous transmit power constraint. Because $\mathcal{P}_{\rm{max}}$ is the maximum module power in one single RAU, the maximum transmit power can be expressed as
\begin{equation*}\label{equ:PMaxExpression}
\begin{split}
\mathcal{P}_{\rm{max}}(t) =&\bigg\{
\begin{array}{lc}
2\mathcal{P}_{\rm{max}},    & \Theta=0\\
\mathcal{P}_{\rm{max}}  ,  & \Theta=1.\\
\end{array}
\end{split}
\end{equation*}
It can be seen that $\mathcal{P}_{\rm{max}}(t)$ is relevant to the specific selected antenna scheme at $t$. The third constraint has been illustrated at the second paragraph of Section \ref{Sec:subsubMathmaticalFormulation}. In addition, to satisfy the demand of delay-sensitive traffic for all $t \in [0,T/2]$, we assume that $\mathcal{P}_{\rm{max}}$ is sufficiently large.

\section{Optimal PAWAS for Various Patterns in Sparse Scattering Scenarios}\label{Sec:PAWASInDSAWGN}

{ In this section, the PAWAS is optimized for delay-insensitive and -sensitive traffic, respectively. Also, their hybrid traffic is discussed. We can describe the traffic pattern by a triple as $(\lambda_{\rm{i}}, \lambda_{\rm{s}}, \tau_{\rm{max}})$. Specifically, for delay-sensitive traffic, $\lambda_{\rm{i}}=0$ and it can be expressed as $(0, \lambda_{\rm{s}}, \tau_{\rm{max}})$. Similarly, for delay-insensitive traffic and hybrid traffic,  they can be expressed as $(\lambda_{\rm{i}}, 0, 0)$ and $(\lambda_{\rm{i}}, \lambda_{\rm{s}}, \tau_{\rm{max}})$, respectively.}

\subsection{Optimal PAWAS for Delay-insensitive Traffic}


{ In sparse scattering scenarios, there are few multi-paths because of little scattering and reflection. In previous works\cite{dong2015power,li2013channel,Dong2014Efficiency}, this propagation channel is considered as AWGN and the channel fading is only determined by large-scale fading. Since the carrier frequency, antenna heights and geography, etc., can be treated as invariant in identical region, their effects on path loss can be omitted and the channel fading coefficient $h_{ij}(t)$ can be simply determined by distance\cite{lu2016location,lu2016locationarXiv,zhang2015optimal}, i.e.}
\begin{equation}\label{equ:LOS_h_ij}
h_{ij}(t) = \big(d_{ij}(t)\big)^{\iota},
\end{equation}
where $\iota$ is the path-loss exponent and $d_{ij}$ (i, j=1, 2) follows
\begin{equation}\label{equ:distancec_ij}
d_{ij}(t) = \sqrt{ d_v^2 + \big( (d_h-(-1)^{i+j}d_{\rm{r}})/2 - (-1)^{j} vt \big)^2 }.
\end{equation}

{ For $(\lambda_{\rm{i}}, 0, 0)$ case, the optimal PAWAS on minimizing the average transmit power in half period can be rewritten as}
\begin{flalign*}
\mathbf{P}_{1\text{-}\rm{A}}:\,\, \min_{\Theta,\mathcal{P}(t)} \,\,\, & \frac{2}{T}\int_{0}^{\frac{T}{2}}\mathcal{P}(t) dt\notag\\
\textrm{s.t.}\,\,\,&\mathcal{P}(t) \leq \mathcal{P}_{\rm{max}}(t),\nonumber\\
& \int_{0}^{T/2}C(t) \geq \lambda_{\rm{i}}T/2.\nonumber\\
\end{flalign*}
{ Note that the first constraint relevant to delay-sensitive traffic in \eqref{equ:MainOptimizeProblemConstraints} is omitted. As previous illustrated, the antenna selection scheme shown in Lemma \ref{lem:PThreshold_DelayConstrainedMDSInAWGN} minimizes instantaneous transmit power. For the average transmit power, it can be minimized by time-domain power allocation, i.e. applying standard Lagrangian method to $\mathbf{P}_{1\text{-}\rm{A}}$. The optimal PAWAS for delay-insensitive traffic can be summarized as follows.}
\begin{prop}\label{prop:NonDelaySensitiveOptimalPowerAndMDS}
For delay-insensitive traffic $(\lambda_{\rm{i}},0,0)$, the optimal time-domain allocated power is
\begin{equation}\label{equ:NDSoptPower}
\mathcal{P}_{\rm{I}}^{*}(t) = \min\left(\mathcal{P}_{\rm{W}}(t) , \mathcal{P}_{\rm{max}}(t)\right),
\end{equation}
and the optimal antenna selection scheme follows Lemma \ref{lem:PThreshold_DelayConstrainedMDSInAWGN}. That is,
\begin{equation*}
\begin{split}
\Theta^* =&\bigg\{
\begin{array}{lc}
0,  & \mathcal{P}_{\rm{I}}^{*}(t) \geq \zeta_{\mathcal{P}}\\
1,  & \mathcal{P}_{\rm{I}}^{*}(t) < \zeta_{\mathcal{P}}.\\
\end{array}
\end{split}
\end{equation*}
$\mathcal{P}_{\rm{W}}(t)$ is the generalized waterfilling power expressed by
\begin{equation}\label{equ:OptimalWaterFillingForNDS}
\mathcal{P}_{\rm{W}}(t) = \bigg[ \frac{\eta\Gamma(\mathcal{P}(t),t) - 1}{\Gamma(\mathcal{P}(t),t)-\eta{\partial \Gamma(\mathcal{P}(t),t)}/{\partial \mathcal{P}(t)}} \bigg]^{+},
\end{equation}
where $\eta$ is the waterfilling coefficient satisfying
\begin{equation*}\label{equ:WaterfillingCoefficient}
\int_0^{T/2} C(t) = \lambda_{\rm{i}}T/2.
\end{equation*}

\end{prop}
\begin{proof}
The proof is shown in Appendix \ref{Appen:PowerAllocationForNDS}.
\end{proof}

\subsection{Optimal PAWAS for Delay-sensitive Traffic}

{ For cases with only delay-sensitive traffic $(0,\lambda_{\rm{s}},\tau_{\rm{max}})$, we assume that when the service queue is not empty, the selected RAUs are power-on and transmit signal, Otherwise, RAUs are power-off. Let $p_{\rm{on}}$ be the power-on probability, we have $p_{\rm{on}}=\rho_{\rm{s}}$\cite{heyman1968optimal}. 
Then, the optimal PAWAS on minimizing the average transmit power in half period can be rewritten as
\begin{flalign*}\label{equ:DSOptimizeProblemInAWGN}
\mathbf{P}_{1\text{-}\rm{B}}:\,\,\min_{\Theta,\mathcal{P}(t)} \,\,\, & \frac{2}{T} \int_{0}^{\frac{T}{2}}p_{\rm{on}} \mathcal{P}(t) dt\notag\\
\textrm{s.t.}\,\,\,&\eqref{equ:QueDelay},\eqref{equ:UniversalCapacity},C(t)=\bar{L}\mu_s.\nonumber\\
\end{flalign*}
Note that the second constraint in \eqref{equ:MainOptimizeProblemConstraints} is ignored, because we assume that $\mathcal{P}_{\rm{max}}$ is sufficiently large. Also, the third constraint relevant to delay-insensitive traffic is omitted. To solve this optimization problem, we analyze the queuing delay of optimal PAWAS first. The result is shown as follows.}
\begin{lem}\label{lem:MaxDelayEfficiency}
In both SIMO and MIMO, $\tau^{*} = \tau_{m}$. $\tau^{*}$ is the optimal queuing delay of the solution to $\mathbf{P}_{1\text{-}\rm{B}}$.
\end{lem}
\begin{proof}
The proof is shown in Appendix \ref{Appen:MaxDelayEfficient}.
\end{proof}

With the results in Lemma \ref{lem:MaxDelayEfficiency} and \eqref{equ:QueDelay}, the optimal service rate is given by
\begin{equation}\label{equ:DemandedRateDelayConstrainedInAWGN}
C^*(t) = \bar{L}\mu_{\rm{s}} = \bar{L}\left(\lambda_{\rm{s}} + {1}/{\tau_{m}}\right).
\end{equation}

Since $\mathbf{P}_{1\text{-}\rm{B}}$ can be solved by minimizing $\mathcal{P}(t)$ point-wise, substitute \eqref{equ:DemandedRateDelayConstrainedInAWGN} into Lemma \ref{lem:PThreshold_DelayConstrainedMDSInAWGN} and \eqref{equ:UniversalCapacity}, we can derive the optimal PAWAS for delay-sensitive traffic as follows.

\begin{prop}\label{prop:DelayConstrainedMDSInAWGN}
For delay-sensitive traffic $(0,\lambda_{\rm{s}},\tau_{\rm{max}})$, the optimal transmit power $\mathcal{P}_{\rm{S}}^{*}(t)$ can be derived by substituting $C^*(t)$ into \eqref{equ:UniversalCapacity}, where $C^*(t)$ is determined by \eqref{equ:DemandedRateDelayConstrainedInAWGN}.

The optimal antenna selection scheme can be expressed as
\begin{equation*}
\begin{split}
\Theta^* =&\bigg\{
\begin{array}{lc}
0,  & C^*(t) \geq \zeta_c\\
1,  & C^*(t) < \zeta_c.\\
\end{array}
\end{split}
\end{equation*}
Practically, when the service queue is empty, we denote $\mathcal{P}_{\rm{S}}^*(t)=0$ and $\Theta^*=1$.
\end{prop}

From Proposition \ref{prop:DelayConstrainedMDSInAWGN}, we can observe that when $C^*(t)$ is sufficiently large, $\Theta^*=0$ holds for $\forall t \in [0,T/2]$, which is summarized as follows.
\begin{cor}\label{cor:MaxZetaMuForMRC}
When $C^*(t) \geq \max\left(\zeta_c^{(1)} , \zeta_c^{(2)}\right)$ for $\forall t \in [0,T/2]$, $\Theta^*\equiv0$. $\zeta_c^{(1)}$ and $\zeta_c^{(2)}$ is the capacity threshold when $x=(d_h-d_{\rm{r}})/2$ and $x=d_h/2$, respectively.
\end{cor}
\begin{proof}
Substitute \eqref{equ:HHExpansionElements} and \eqref{equ:LOS_h_ij} into \eqref{equ:MDSRThreshold}, we have
\begin{equation*}
\frac{\partial{\zeta_c}}{\partial x}\bigg|_{x = (d_h-d_{\rm{r}})/2} = 0
\end{equation*}
and
\begin{equation*}
\lim_{x\to {(d_h/2)}^{-}} \frac{\partial{\zeta_c}}{\partial x} > 0,
\end{equation*}
where $x=vt$. That is, the capacity threshold $\zeta_c$ achieves its local maximum at $x=(d_h-d_{\rm{r}})/2$ and $x=d_h/2$, and the global maximum is $\max\left(\zeta_c^{(1)} , \zeta_c^{(2)}\right)$. According to Proposition \ref{prop:DelayConstrainedMDSInAWGN}, $\Theta^*\equiv0$ for $\forall t \in [0,T/2]$ when $C^*(t) \geq \max\left(\zeta_c^{(1)} , \zeta_c^{(2)}\right)$.
\end{proof}


\begin{figure}[htbp]
\centering
\subfigure[]{ \label{Fig:HybridSamplePower}
\includegraphics[width=0.3\columnwidth]{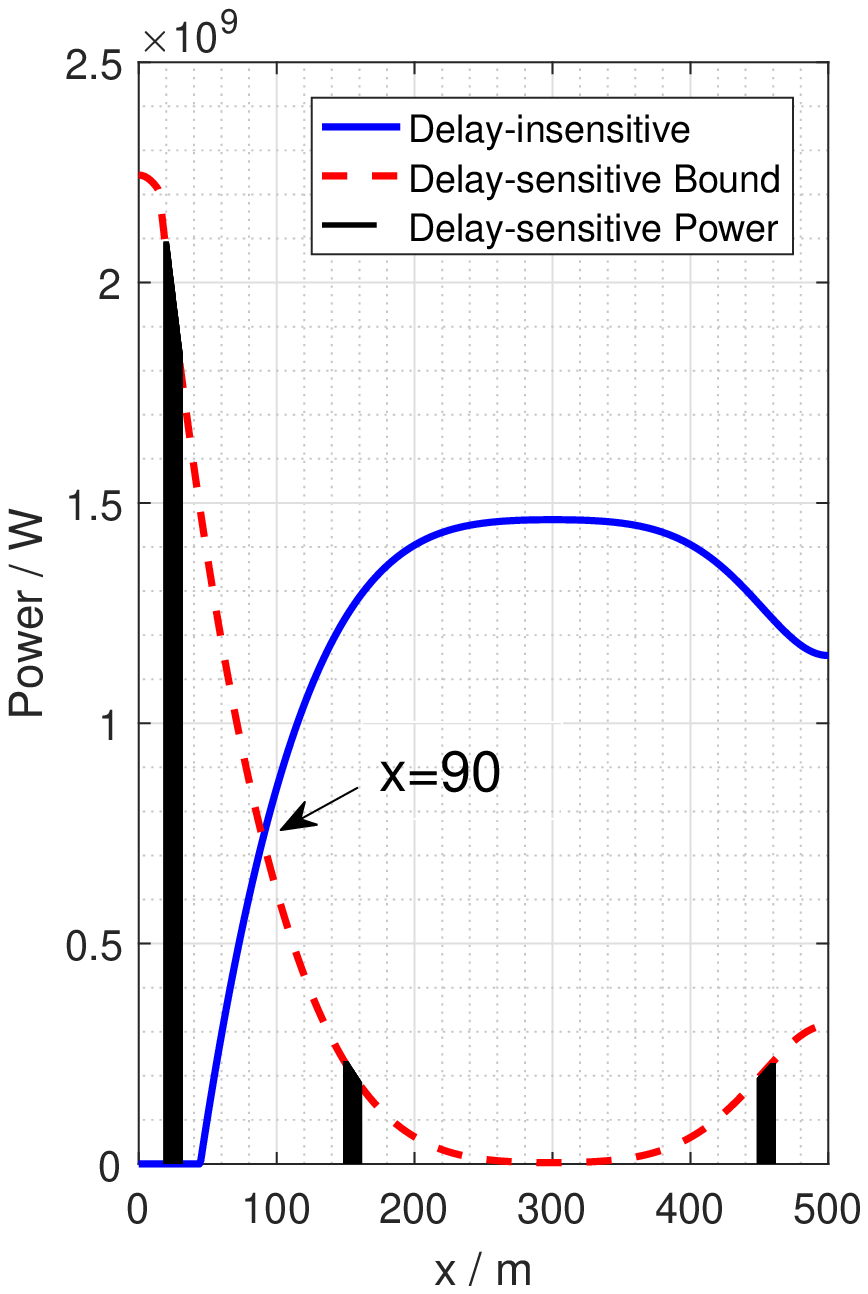}}
\subfigure[]{ \label{Fig:HybridSampleCapacity}
\includegraphics[width=0.3\columnwidth]{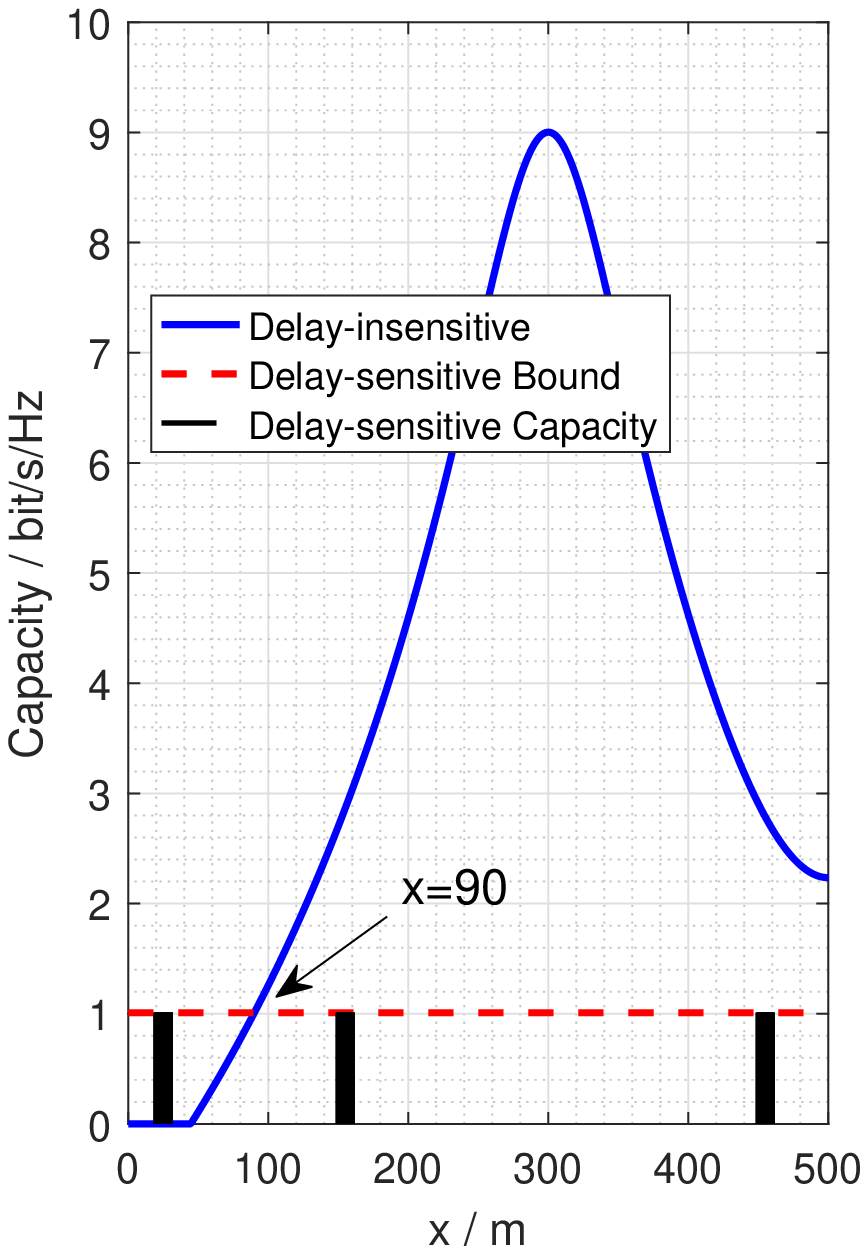}}
\caption{The power allocations and capacities for delay-insensitive and delay-sensitive traffic, where $\lambda_{\rm{i}}=40~\rm{s}^{-1}$, $\lambda_{\rm{s}}=1~\rm{s}^{-1}$ and $\tau_{\rm{max}} = 110~\rm{ms}$. In simulations, $d_{\rm{r}}$=400 m, $d_h$=1000 m, $d_v$=100 m, $\iota$=3.8 and $\bar{L}=0.1$ bit/Hz.} \label{Fig:HybridSample} \label{Fig:HybridSample}
\end{figure}

\subsection{Optimal PAWAS for Hybrid Traffic}

{ To derive the optimal PAWAS for hybrid traffic $(\lambda_{\rm{i}},\lambda_{\rm{s}},\tau_{\rm{max}})$, we take $(40,1,110)$ for example.} As shown in Fig. \ref{Fig:HybridSample}, the optimal power allocation for delay-sensitive and -insensitive traffic are depicted, separately. The solid lines in Fig. \ref{Fig:HybridSamplePower} and Fig. \ref{Fig:HybridSampleCapacity} denote the allocated power and capacity for the delay-insensitive part $(40,0,0)$, respectively. Since the service process to the delay-sensitive part $(0,1,110)$ is probabilistic, the scattered vertical black bars denote the relevant stochastically allocated power and capacity, and the dashed lines are corresponding power and capacity bounds.

We extend the theoretical results of delay-insensitive traffic to hybrid traffic. Let $C(t)$ be the channel capacity relevant to the delay-insensitive part $(\lambda_{\rm{i}},0,0)$. As shown in Fig. \ref{Fig:HybridSampleCapacity}, for hybrid traffic $(40,1,110)$, servicing delay-insensitive traffic occupies $C(t)$ in the service capacity of delay-insensitive traffic in $t\in[0,90/v]$, whereas in $t\in[0,90/v]$, it occupies $\bar{L}\mu_{\rm{s}}$. Denote the cumulative occupied capacity by delay-sensitive traffic in half period as $S_{\rm{op}}$. For general hybrid traffic cases, since the busy probability corresponding to delay-sensitive traffic is $p_{\rm{on}}$, the expectation of $S_{\rm{op}}$ is
\begin{equation*}
\mathbb{E}(S_{\rm{op}}) = \int_{0}^{T/2} p_{\rm{on}}\left(\min\left(\bar{L}\mu_{\rm{s}} , C(t)\right)\right) dt.
\end{equation*}

To make up the commutative capacity occupied by delay-sensitive traffic, the waterfilling coefficient for hybrid traffic is supposed to be
\begin{equation}\label{equ:WaterFillingCoefficientForHybrid}
\begin{split}
&\hat{\eta} = \argmine_{\hat{\eta}} \Big\{ \int_0^{T/2} C(t)dt - \mathbb{E}(S_{\rm{op}}) = \bar{L}\lambda_{\rm{i}} T/2 \Big\}\\
&=\argmine_{\hat{\eta}} \Big\{ \int_0^{T/2} \Big(C(t) - p_{\rm{on}}\min\big(\bar{L}\mu_{\rm{s}} , C(t)\big)\Big)dt = \bar{L} \lambda_{\rm{i}} T/2 \Big\}.
\end{split}
\end{equation}

Note that the waterfilling coefficient $\hat{\eta}$ shown in \eqref{equ:WaterFillingCoefficientForHybrid} is derived by the expectation of $S_{\rm{op}}$. This is because the arrival of delay-sensitive traffic is stochastic. With this mode, the delay-insensitive traffic can be serviced in one period by probability, and infinite queuing delay is avoided. The optimal PAWAS for hybrid traffic can be summarized as follows.
\begin{prop}\label{prop:HybridOptimalPowerAndMDS}
For hybrid traffic $(\lambda_{\rm{i}},\lambda_{\rm{s}},\tau_{\rm{max}})$, the optimal power allocation is
\begin{equation}\label{equ:HybridInstantaneousPower}
\mathcal{P}_{\rm{H}}^*(t) = \max(\hat{\mathcal{P}}_{\rm{I}}^*(t) , \mathcal{P}_{\rm{S}}^*(t)),
\end{equation}
where $\hat{\mathcal{P}}_{\rm{I}}^*(t)$ is determined by applying \eqref{equ:WaterFillingCoefficientForHybrid} to \eqref{equ:NDSoptPower}, and $\mathcal{P}_{\rm{S}}^*(t)$ is the optimal allocated power for $(0,\lambda_{\rm{s}},\tau_{\rm{max}})$.

The optimal antenna selection scheme is
\begin{equation*}
\begin{split}
\Theta^* =&\bigg\{
\begin{array}{lc}
0,  & \mathcal{P}_{\rm{H}}^*(t) \geq \zeta_{\mathcal{P}}\\
1,  & \mathcal{P}_{\rm{H}}^*(t) < \zeta_{\mathcal{P}}.\\
\end{array}
\end{split}
\end{equation*}
\end{prop}

For some specific cases with high delay-insensitive traffic demand, we have following corollary.

\begin{cor}\label{cor:TaumNotCare}
Hybrid traffics with same $(\lambda_{\rm{i}} + \lambda_{\rm{s}})$ have identical $\hat{\mathcal{P}}^*_{\rm{I}}(t)$ and $C(t)$ for arbitrary $\tau_{\rm{max}}$, when $C(t) \geq \bar{L}\big(\lambda_{\rm{s}}+{1}/{\tau_{\rm{max}}}\big)$ holds in $\forall t\in[0,T/2]$. In this case, $\hat{\eta}$ is determined by
\begin{equation}\label{equ:SpecificHybridCt}
\hat{\eta} = \argmine_{\hat{\eta}} \Big\{\int_0^{T/2} C(t) = \bar{L} (\lambda_{\rm{i}} + \lambda_{\rm{s}} )T/2\Big\}.
\end{equation}
\end{cor}
\begin{proof}
Substitute $C(t) \geq \bar{L}\mu_{\rm{s}} = \bar{L}\left(\lambda_{\rm{s}}+{1}/{\tau_{\rm{max}}}\right)$ into \eqref{equ:WaterFillingCoefficientForHybrid}, \eqref{equ:SpecificHybridCt} can be easily derived. Then, traffics with the same value of $(\lambda_{\rm{i}} + \lambda_{\rm{s}})$ have the same value of $\hat{\eta}$, resulting in identical $\hat{\mathcal{P}}_{I}(t)$ and $C(t)$.
\end{proof}


Based on previous theoretical results, the optimal PAWAS for hybrid traffic can be solved by Algorithm \ref{alg:PowerAllocationAndMDSInAWGN}. Specifically, when $\lambda_{\rm{s}}=0~s^{-1}$, hybrid traffic degrades to delay-insensitive traffic, and Algorithm \ref{alg:PowerAllocationAndMDSInAWGN} is the solution to Proposition \ref{prop:NonDelaySensitiveOptimalPowerAndMDS}. Similarly, when $\lambda_{\rm{i}} = 0~\rm{s}^{-1}$, it degrades to delay-sensitive traffic, and Algorithm \ref{alg:PowerAllocationAndMDSInAWGN} solves Proposition \ref{prop:DelayConstrainedMDSInAWGN}.

\begin{algorithm}[htbp]
    \caption{Optimal PAWAS}\label{alg:PowerAllocationAndMDSInAWGN}
    \begin{algorithmic}[1]

    	\item \textbf{Initialize} ~~\\          
        Physical framework parameters: $d_{\rm{r}}$, $d_{\rm{h}}$, $d_{\rm{v}}$ and $v$. Period $T = {d_{\rm{h}}}/{v}$;\\
        Input queuing parameters: $\lambda_{\rm{i}}$, $\lambda_{\rm{s}}$, $\tau_{\rm{m}}$;\\
        Initialize iteration parameters: $\hat{\eta}=\hat{\eta}_{\rm{min}}=0$, $\hat{\eta}_{\rm{max}}=1$ and $C(t)=0$ for $\forall t \in [0,T/2]$;\\
        Set the precision $\epsilon=10^{-3}$.\\

        \WHILE {$\int_0^{T/2} \Big(C(t) - p_{\rm{on}}\min\big(\bar{L}\mu_{\rm{s}} , C(t)\big)\Big)dt \leq \frac{\bar{L}\lambda_{\rm{i}} T}{2}$}\label{code:FindMaxEta}
        	\STATE $\hat{\eta}_{\rm{max}} = 10\hat{\eta}_{\rm{max}}$;\\
            \STATE $\hat{\eta}=\hat{\eta}_{\rm{max}}$;\\
            \STATE $\mathcal{P}_{\rm{H}}(t)\leftarrow$ Substitute $\hat{\eta}_{\rm{max}}$ into \eqref{equ:OptimalWaterFillingForNDS} and \eqref{equ:HybridInstantaneousPower};\label{line:3}
            \STATE $C(t) \leftarrow \log_2\left( 1 + \Gamma(\mathcal{P}_{\rm{H}}(t),t) \mathcal{P}_{\rm{H}}(t) \right)$.\label{line:1}
        \ENDWHILE \\

        \WHILE {$\int_0^{T/2} \Big(C(t) - p_{\rm{on}}\max\big(\bar{L}\mu_{\rm{s}} , C(t)\big)\Big)dt - \frac{\bar{L}\lambda_{\rm{i}} T}{2} \geq \epsilon$}\label{code:FindAccurateEta}
        	\STATE $\hat{\eta} = (\hat{\eta}_{\rm{min}} + \hat{\eta}_{\rm{max}})/2$;\\
            \STATE $\mathcal{P}_{\rm{H}}(t)\leftarrow$ Substitute $\hat{\eta}$ into \eqref{equ:OptimalWaterFillingForNDS} and \eqref{equ:HybridInstantaneousPower};\\
            \IF  {$\mathcal{P}_{\rm{H}}(t) \geq \zeta_{\mathcal{P}}$}\label{line:4}
				    \STATE $\Theta=0$;
			    \ELSE
			        \STATE $\Theta=1$;
			    \ENDIF
            \STATE $C(t) \leftarrow \log_2\left( 1 + \Gamma(\mathcal{P}_{\rm{H}}(t),t) \mathcal{P}_{\rm{H}}(t) \right)$.\label{line:2}
            \IF  {$\int_0^{T/2} \Big(C(t) - p_{\rm{on}}\min\big(\bar{L}\mu_{\rm{s}} , C(t)\big)\Big)dt \geq \frac{\bar{L}\lambda_{\rm{i}} T}{2}$}
            \STATE $\hat{\eta}_{\rm{max}} = \hat{\eta}$;
            \ELSE
            \STATE $\hat{\eta}_{\rm{min}} = \hat{\eta}$;
            \ENDIF
        \ENDWHILE \\

       	
       \textbf{Output} ~Optimal $\mathcal{P}^*_{\rm{H}}(t)$ and $\Theta^*$.\\          
    \end{algorithmic}
\end{algorithm}

\section{Optimal PAWAS in Rich Scattering Scenarios}\label{Sec:PAWASInHybridNakagami}

{ This section analyzes the optimal PAWAS with severe small-scale fading under rich scattering scenarios. Due to the high velocity of train, HSR channel experiences fast fading, and tracking of instantaneous small-scale fading would be impossible. According to \cite{zhang2015optimal}, the statistics of small-scale fading in one area can usually be determined in advance and hold for a long time, which can be well modeled by Nakagami-$m$ distribution\cite{li2013channel,dong2012varepsilon,liu2014novel}. Specifically, previously discussed sparse scattering scenarios are the special case of $m\rightarrow\infty$. In \cite{liu2014novel,chen2016smart}, we proposed a $m$ factor estimator suitable to fast time-varying HSR channels and implemented the real-time estimation by hardware. The estimated $m$ can be used to design the optimal PAWAS for rich scattering scenarios.

In addition, high velocity of train leads to small coherence time. For instance, for high-speed trains currently tested in China, the velocity is 486 km/h and the coherence time would be approximately 0.26 ms\cite{liu2012position}, which is small compared with our considered $\tau_{\rm{max}}$ and $T/2$.} Then, the instantaneous channel capacity of MIMO and SIMO, defined as $C_{\rm{M}}(m,t)$ and $C_{\rm{S}}(m,t)$, is ergodic within $\tau_{\rm{max}}$ or $T/2$. Let $\mathbb{C}_{\rm{M}}(m,t)$ and $\mathbb{C}_{\rm{S}}(m,t)$ be the corresponding ergodic capacity, we have
\begin{equation*}\label{equ:ServiceProvidedWithSmallScaleFaing}
\begin{cases}
\int_{0}^{\kappa} C_{\rm{M}}(m,t) dt = \mathbb{C}_{\rm{M}}(m,t) \kappa\\
\int_{0}^{\kappa} C_{\rm{S}}(m,t) dt = \mathbb{C}_{\rm{S}}(m,t) \kappa,
\end{cases}
\end{equation*}
where $\kappa$ is $\tau_{\rm{max}}$ or $T/2$. That is, the commutative channel capacity under Nakagami-$m$ fading can be characterized by corresponding ergodic capacity, which is determined as follows.

\begin{thm}\label{thm:HighSNRNakagamiCapacity}
The ergodic capacity of MIMO and SIMO can be expressed as
\begin{equation}\label{equ:MIMOCapacityWithSmallScaleFaing}
\mathbb{C}_{\rm{M}}(m,t) = C_{\rm{M}}(t) + \frac{2\big(\psi(m)-\ln(m)\big)}{\ln2}
\end{equation}
and
\begin{equation}\label{equ:MRCCapacityWithSmallScaleFaing}
\mathbb{C}_{\rm{S}}(m,t) = C_{\rm{S}}(t) + \frac{\psi(m)-\ln(m)}{\ln2},
\end{equation}
respectively. $\psi(m)$ is the digamma function\cite[(8.365.4)]{gradshteyn1994table}.
\end{thm}
\begin{proof}
The proof is shown in Appendix \ref{Appen:NakagamiCapacity}.
\end{proof}

Note that Theorem \ref{thm:HighSNRNakagamiCapacity} is derived in high SNR case. This is because we are interested in supporting high transmission demand and small delay constraint\cite{zhang2015optimal}, and hence we restrict our analysis to high SNR case. The derived results are quite informative. This clearly indicates the individual effects of large-scale and small-scale fading on ergodic capacity. That is, the first parts of \eqref{equ:MIMOCapacityWithSmallScaleFaing} and \eqref{equ:MRCCapacityWithSmallScaleFaing} account for the large-scale fading, while the second parts explain the ergodic capacity loss caused by small-scale fading.

Now, we shall analyze the optimal PAWAS for hybrid traffic with Nakagami-$m$ fading, since PAWAS for delay-sensitive and -insensitive traffic are the special cases. Note that the effects of small-scale fading is only determined by $m$. Substitute \eqref{equ:MIMOCapacityWithSmallScaleFaing} and \eqref{equ:MRCCapacityWithSmallScaleFaing} into \eqref{equ:LagrangianFunction}, it can be derived that the $\mathcal{P}_{\rm{W}}(t)$ with Nakagami-$m$ fading can also be expressed as \eqref{equ:OptimalWaterFillingForNDS}. However, since the ergodic capacity loss caused by small-scale fading in MIMO doubles that in SIMO, it's hard to derive explicit $\zeta_{c}$ and $\zeta_{\mathcal{P}}$ in Nakagami-$m$ case. { Therefore, we rewrite the relevant antenna selection scheme as}
\begin{equation*}
\begin{split}
\Theta =&\bigg\{
\begin{array}{lc}
0,    & \mathcal{P}_{\rm{H}}(t)|_{\Theta=1} \geq \mathcal{P}_{\rm{H}}(t)|_{\Theta=0}\\
1,  & \mathcal{P}_{\rm{H}}(t)|_{\Theta=1} < \mathcal{P}_{\rm{H}}(t)|_{\Theta=0},\\
\end{array}
\end{split}
\end{equation*}
where $\mathcal{P}_{\rm{H}}(t)|_{\Theta=1}$ and $\mathcal{P}_{\rm{H}}(t)|_{\Theta=0}$ is the transmit power derived by assuming the selected antenna scheme is SIMO or MIMO, respectively. In this way, the instantaneous transmit power is minimized. The waterfilling coefficient $\hat{\eta}$ can be expressed as
\begin{equation}\label{equ:WaterFillingCoefficientForHybridInNakagami}
\begin{split}
&\hat{\eta}=\argmine_{\hat{\eta}} \Big\{ \int_0^{T/2} \Big(\mathbb{C}(t) - p_{\rm{on}}\min\big(\bar{L}\mu_{\rm{s}} , \mathbb{C}(t)\big)\Big)dt = \bar{L} \lambda_{\rm{i}} T/2 \Big\}.
\end{split}
\end{equation}
Note that $\mathbb{C}(t)$ in \eqref{equ:WaterFillingCoefficientForHybridInNakagami} is the ergodic channel capacity corresponding to the delay-insensitive part, i.e. $(\lambda_{\rm{i}},0,0)$. Then, the optimal PAWAS for hybrid traffic with Nakagami-$m$ fading can be summarized as follows.

\begin{prop}\label{prop:HybridOptimalPowerAndMDSInNakagami}
For hybrid traffic $(\lambda_{\rm{i}},\lambda_{\rm{s}},\tau_{\rm{max}})$ with Nakagami-$m$ fading, the optimal power allocation is
\begin{equation}\label{equ:HybridInstantaneousPowerInNakagami}
\mathcal{P}^*_{\rm{H}}(t) = \max\left(\hat{\mathcal{P}}^*_{\rm{I}}(t) , \mathcal{P}^*_{\rm{S}}(t)\right),
\end{equation}
where $\hat{\mathcal{P}}^*_{\rm{I}}(t)$ is determined by applying \eqref{equ:WaterFillingCoefficientForHybridInNakagami} to \eqref{equ:NDSoptPower}, and $\mathcal{P}^*_{\rm{S}}(t)$ is the optimal allocated power for the delay-sensitive part, i.e. $(0,\lambda_{\rm{s}},\tau_{\rm{max}})$.

The optimal antenna selection scheme is
\begin{equation*}
\begin{split}
\Theta^* =&\bigg\{
\begin{array}{lc}
0,    & \mathcal{P}_{\rm{H}}(t)|_{\Theta=1} \geq \mathcal{P}_{\rm{H}}(t)|_{\Theta=0}\\
1,  & \mathcal{P}_{\rm{H}}(t)|_{\Theta=1} < \mathcal{P}_{\rm{H}}(t)|_{\Theta=0}.\\
\end{array}
\end{split}
\end{equation*}
\end{prop}

Accordingly, Proposition \ref{prop:HybridOptimalPowerAndMDSInNakagami} can be solved with few amendments on Algorithm \ref{alg:PowerAllocationAndMDSInAWGN}, which are listed as follows
\begin{enumerate}
\item In line \ref{line:1} and \ref{line:2}, $\log_2\left( 1 + \Gamma(\mathcal{P}_{\rm{H}}(t),t) \mathcal{P}_{\rm{H}}(t)\right)$ is replaced by \eqref{equ:MIMOCapacityWithSmallScaleFaing} or \eqref{equ:MRCCapacityWithSmallScaleFaing};
\item In line \ref{line:3}, $\mathcal{P}^*_{\rm{H}}(t)$ is derived by substituting $\hat{\eta}_{\rm{max}}$ into \eqref{equ:OptimalWaterFillingForNDS} and \eqref{equ:HybridInstantaneousPowerInNakagami};
\item In line \ref{line:4}, replace $\mathcal{P}^*_{\rm{H}}(t) \geq \zeta_{\mathcal{P}}$ with $\mathcal{P}_{\rm{H}}(t)|_{\Theta=1} \geq \mathcal{P}_{\rm{H}}(t)|_{\Theta=0}$.
\end{enumerate}

\section{Numerical Results}\label{Sec:NumericalRes}

In this section, numerical results are presented to show the validity of our theoretical results and provide more insights on the effectiveness of our proposed PAWAS. In simulations, the system parameters shown in Fig. \ref{Fig:MultipleAntennaModel} are $d_{\rm{r}}$=400 m, $d_v$=100 m and $d_h$=1000 m. The velocity of train $v$ and the path-loss exponent $\iota$ is set as 500 km/h and 3.8, respectively. The normalized average package length $\bar{L}$ is assumed to be 0.01 bit/Hz.

\subsection{Delay-sensitive Traffic in Sparse Scattering Scenarios}\label{Simu:DSInAWGN}

\begin{figure}[tbp]
\centering
\subfigure[$\lambda_{\rm{s}} = 800~\rm{s}^{-1}$, $\mathcal{P}_{\rm{max}}=110$ dB]{ \label{Fig:DS800Pmax}
\includegraphics[width=0.23\columnwidth]{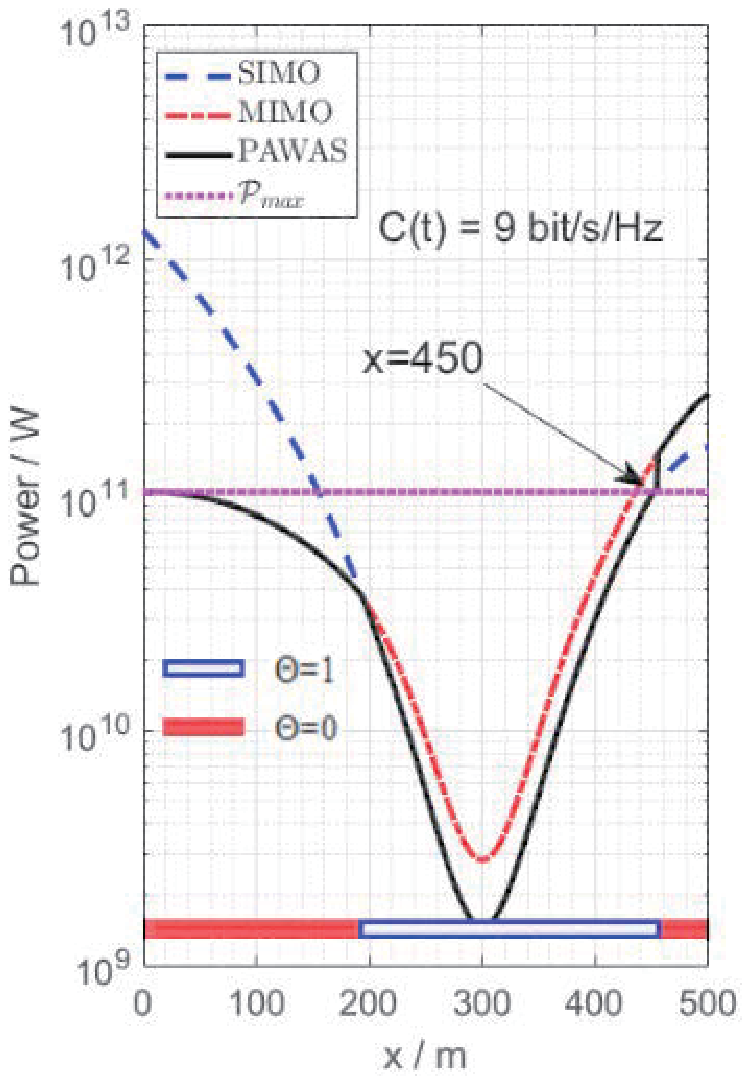}}
\subfigure[$\lambda_{\rm{s}} = 800~\rm{s}^{-1}$] { \label{Fig:DS800}
\includegraphics[width=0.23\columnwidth]{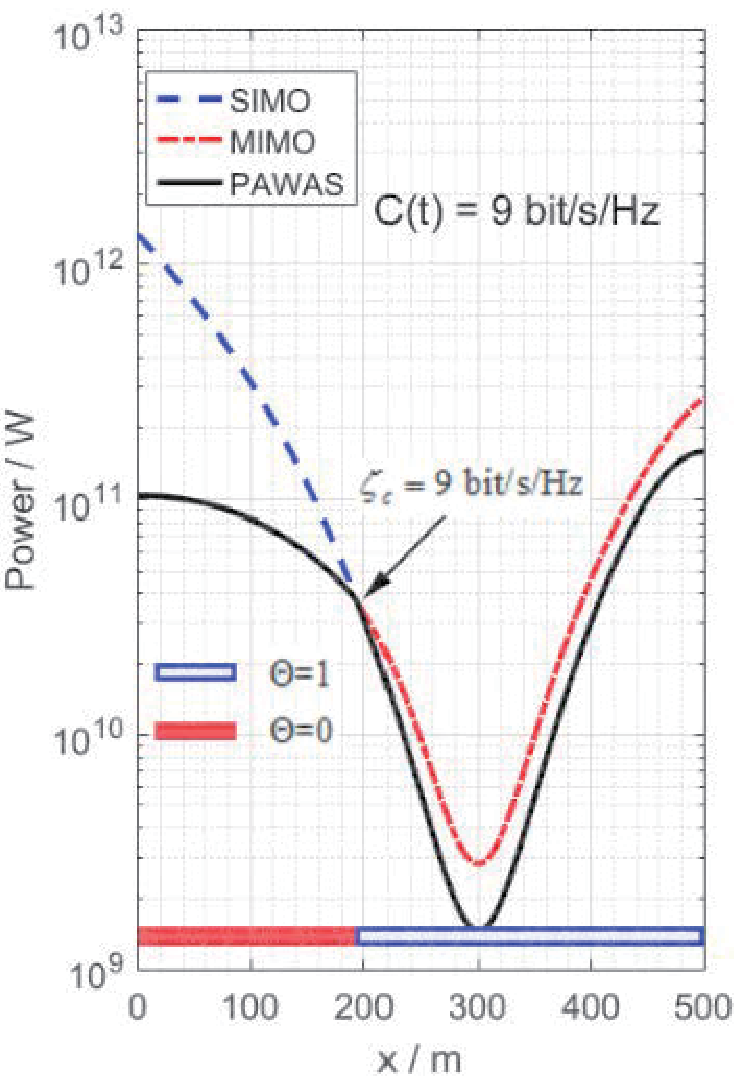}}
\subfigure[$\lambda_{\rm{s}} = 1000~\rm{s}^{-1}$] { \label{Fig:DS1000}
\includegraphics[width=0.23\columnwidth]{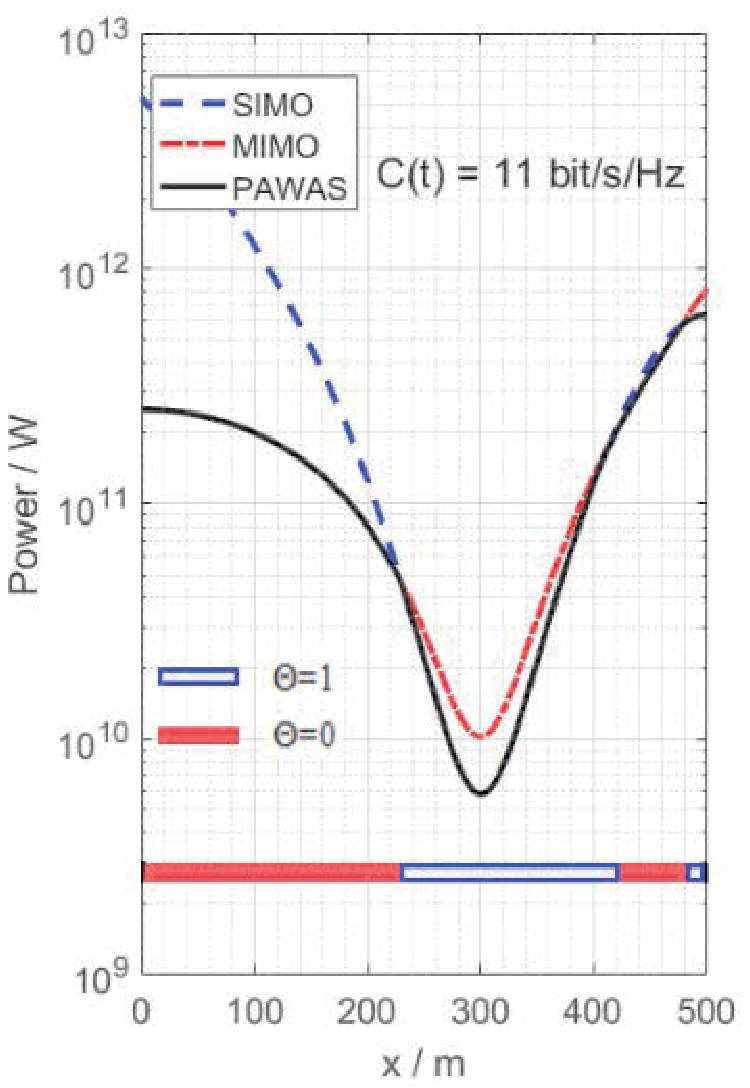}}
\subfigure[$\lambda_{\rm{s}} = 1500~\rm{s}^{-1}$] { \label{Fig:DS1500}
\includegraphics[width=0.23\columnwidth]{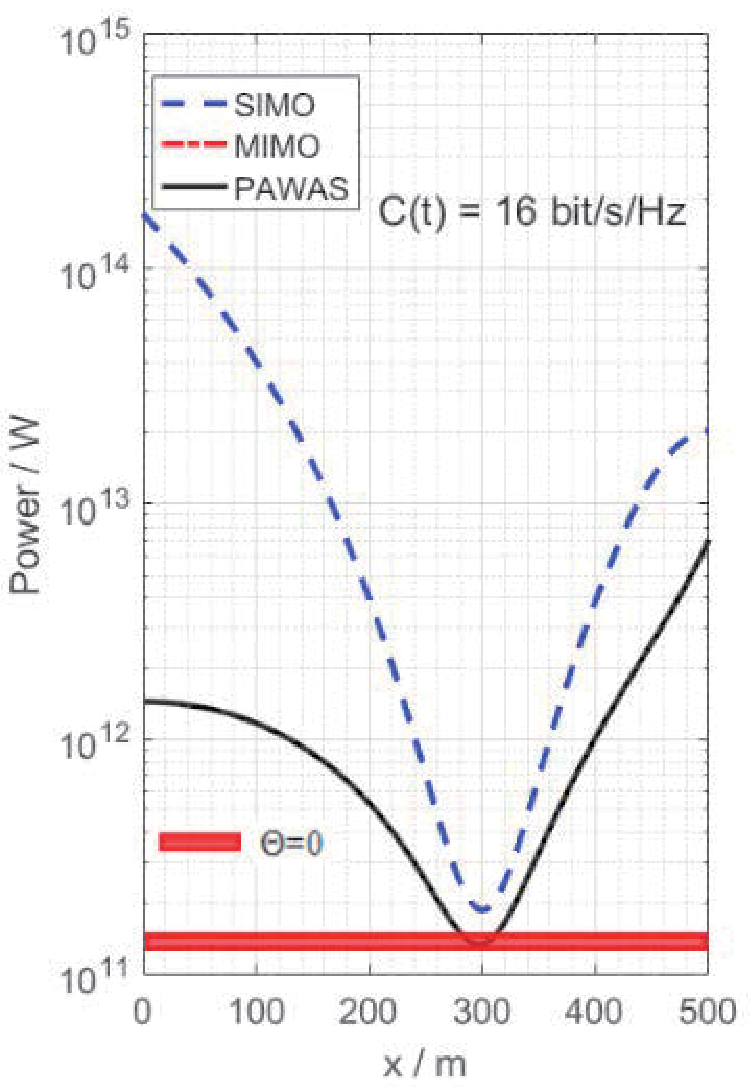}}
\caption{The optimal PAWAS for $(0,\lambda_{\rm{s}},10)$. In (a), the maximum power constraint in each individual RAUs is 110 dB, while in (b)-(d), $\mathcal{P}_{\rm{max}}$ are assumed to be sufficiently large. $\lambda_{\rm{s}}$ in (a) and (b) is 800 $\rm{s}^{-1}$. In (c) and (d), $\lambda_{\rm{s}}=$ 1000 $\rm{s}^{-1}$ and 1500 $\rm{s}^{-1}$, respectively. In all these sub-figures, $\tau_{\rm{max}}=10$ ms. The colored bar at the bottom of each sub-figures demonstrates the selected antenna scheme.} \label{Fig:DSFrom800to1500}
\end{figure}

For delay-sensitive traffic, the optimal PAWAS is depicted in Fig. \ref{Fig:DSFrom800to1500}. The corresponding $\tau_{\rm{max}}$ is 10 ms, and $\rm{RAU}_2$ in Fig. \ref{Fig:MultipleAntennaModel} is located at $x=500$ m, i.e. the right margin of each sub-figures. In Fig. \ref{Fig:DS800Pmax}, the maximum power constraint $\mathcal{P}_{\rm{max}}$ in each individual RAUs is 110 dBW, while in Fig. \ref{Fig:DS800}-\ref{Fig:DS1500}, $\mathcal{P}_{\rm{max}}$ are assumed to be sufficiently large. The arrival rate $\lambda_{\rm{s}}$ in Fig. \ref{Fig:DS800Pmax} and Fig. \ref{Fig:DS800} is 800 $\rm{s}^{-1}$, and in Fig. \ref{Fig:DS1000} and Fig. \ref{Fig:DS1500}, $\lambda_{\rm{s}}=$ 1000 $\rm{s}^{-1}$ and 1500 $\rm{s}^{-1}$, respectively. The demanded service rates $C(t)$ are depicted in each sub-figures, where the maximum delay constraint $\tau_{\rm{max}}=10$ ms.

In Fig. \ref{Fig:DS800Pmax}, it can be observed that even though SIMO is more energy efficient in $x\in[450,500]$, MIMO is selected due to the limited $\mathcal{P}_{\rm{max}}$. Notice that in Fig. \ref{Fig:DS800}, MIMO is switched to SIMO at $x=200$. In this position, the theoretical $\zeta_c$ derived from Lemma \ref{lem:PThreshold_DelayConstrainedMDSInAWGN} is 9 bit/s/Hz, which agrees with the simulated service rate and verifies the validity of Lemma \ref{lem:PThreshold_DelayConstrainedMDSInAWGN}. Observe Fig. \ref{Fig:DS800Pmax} and Fig. \ref{Fig:DS800}, it can be seen that MIMO is selected at positions near to $x=0$. This is because the difference between $\alpha_1$ and $\alpha_2$ shown in \eqref{equ:MDSPThreshold} and \eqref{equ:MDSRThreshold} is small, resulting in relatively small $\zeta_\mathcal{P}$ and $\zeta_c$. A more intuitive explanation is that since both adjacent RAUs are reachable to receive antennas in these positions, MIMO outperforms SIMO decreases of multiplex gain. By contrast, in positions near to $\rm{RAU}_2$ located at $x=500$, the path attenuation from $\rm{RAU}_1$ is extremely high. Therefore, allocating full power to $\rm{RAU}_2$ is more energy efficient than splitting power evenly to both $\rm{RAU}_1$ and $\rm{RAU}_2$.

Compare the allocated power of PAWAS with that of MIMO and SIMO in Fig. \ref{Fig:DS800}-\ref{Fig:DS1500}, it can be seen that our proposed PAWAS always select the antenna scheme with less transmit power. That is, the average transmit power for delay-sensitive traffic can be efficiently minimized by PAWAS. Also, the allocated power of PAWAS can be considered as a generalization of channel-inversion. As shown in Fig. \ref{Fig:MultipleAntennaModel}, the channel gain achieves its maximum at $x=300$, where the distance between $\rm{MR}_2$ and $\rm{RAU}_2$ is minimized. Therefore, system allocates minimum transmit power to $x=300$. At other positions, the farther $\rm{MR}_2$ to $\rm{RAU}_2$, the lower the channel gain, and hence the higher the allocated transmit power. Also, in Fig \ref{Fig:DS800}-\ref{Fig:DS1500}, it can be observed that the higher the $\lambda_{s}$, the more positions that MIMO covers. This is because according to \eqref{equ:DemandedRateDelayConstrainedInAWGN}, the demanded capacity is proportional to $\lambda_{s}$. Since the capacity threshold $\zeta_c$ remains unchanged, increasing capacity demand leads system to chose MIMO in more positions. Especially, for case with $\lambda_{\rm{s}}$=1500 $\rm{s}^{-1}$, MIMO is selected for all $x \in [0,500]$. This is consistent with Corollary \ref{cor:MaxZetaMuForMRC}, because the demanded capacity $C(t)=$16 bit/s/Hz exceeds the theoretical maximum $\zeta_c=13.7$ bit/s/Hz.


\subsection{Delay-insensitive Traffic in  Sparse Scattering Scenarios}\label{sec:DelayInsAWGN}

\begin{figure}[tbp]
\centering
\subfigure[]{ \label{Fig:DelayInsAvePowPA}
\includegraphics[width=0.3\columnwidth]{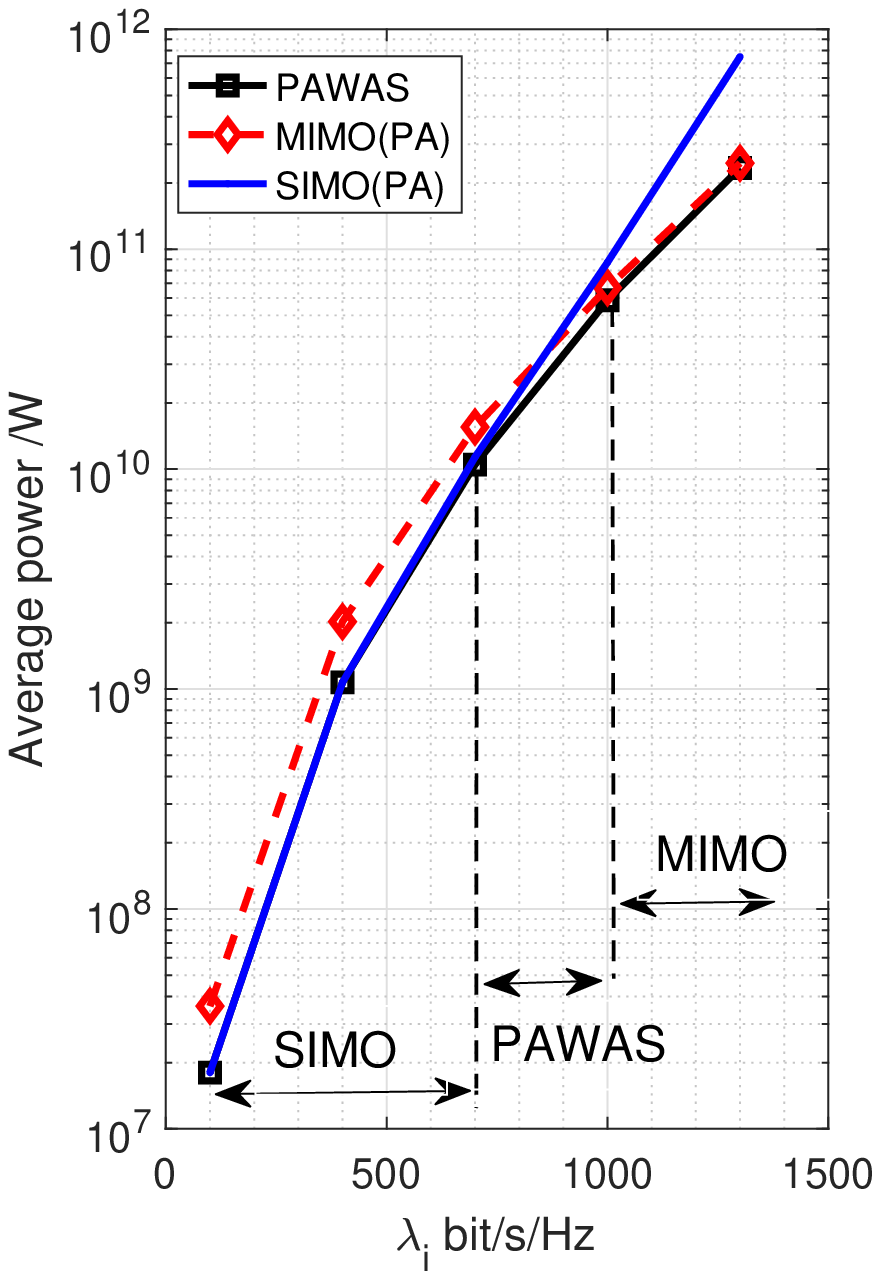}}
\subfigure[]{ \label{Fig:DelayInsAvePowEA}
\includegraphics[width=0.3\columnwidth]{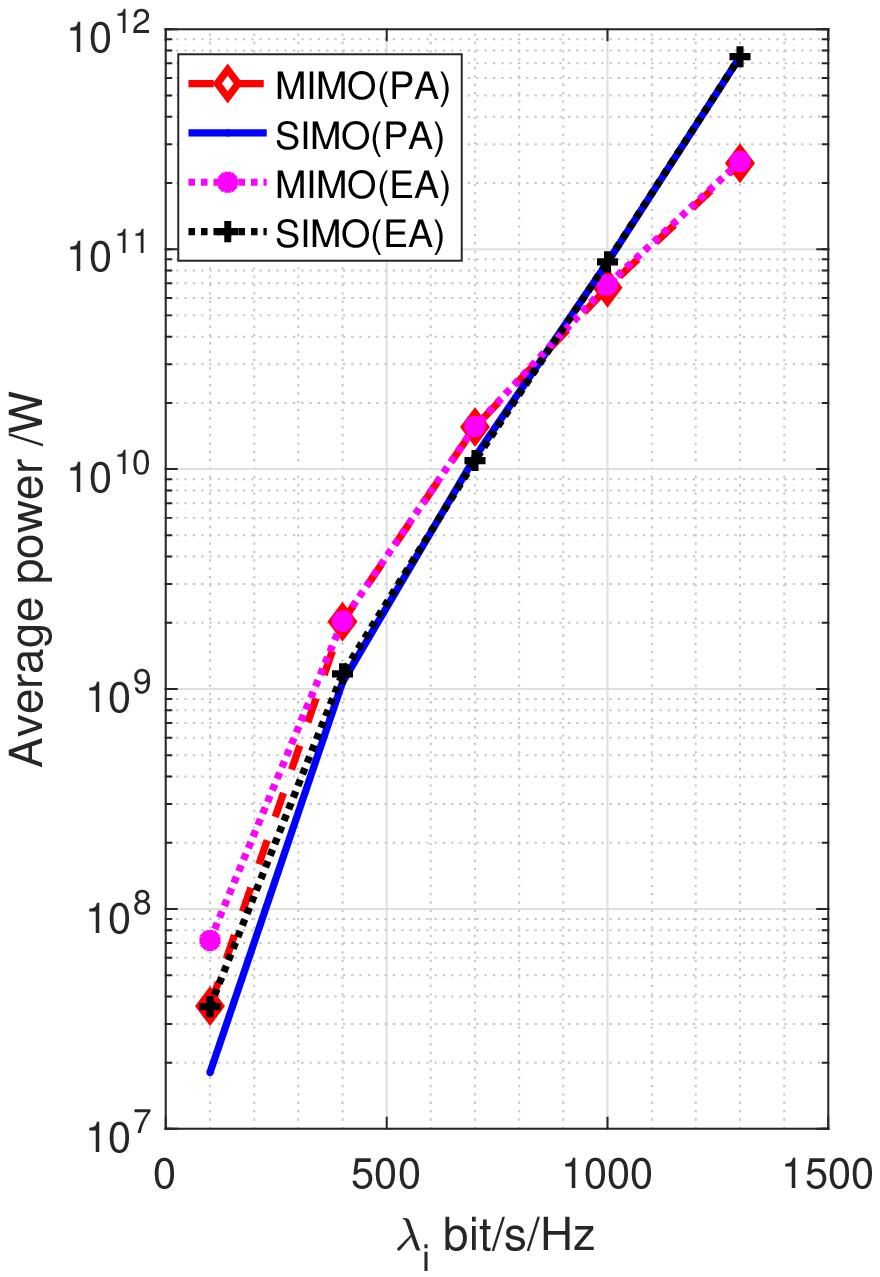}}
\caption{(a) compares the average power between PAWAS, MIMO with time-domain power allocation and SIMO with time-domain power allocation, and (b) compares time-domain power allocation and even power allocation with respect to MIMO and SIMO. The serviced traffic is $(\lambda_{\rm{i}},0,0)$, where $\lambda_{\rm{i}}$ varies between [100,1300] bit/s/Hz.} \label{Fig:DelayInsAvePow}
\end{figure}

Since PAWAS for delay-insensitive traffic is one of the special cases of that for hybrid traffic, we can discuss its performance by the simulations of hybrid traffic in following sub-section. In this part, we only verify the effectiveness of our proposed PAWAS on minimizing average transmit power, when system serves delay-insensitive traffic. Fig. \ref{Fig:DelayInsAvePowPA} compares the average power between PAWAS, MIMO and SIMO with time-domain power allocation. In addition, Fig. \ref{Fig:DelayInsAvePowEA} compares time-domain power allocation and even power allocation with respect to MIMO and SIMO. The arrival rate of delay-insensitive traffic $\lambda_{\rm{i}}$ varies between [100,1300] bit/s/Hz.

In Fig. \ref{Fig:DelayInsAvePowPA}, one can observed that our proposed PAWAS consumes least average transmit power for arbitrary $\lambda_{\rm{i}}$ in these five simulated power-allocation methods, since Fig. \ref{Fig:DelayInsAvePowEA} shows that the even power allocation method in MIMO and SIMO consumes more average transmit power than the corresponding time-domain power allocation method. Similarly, in sub-section \ref{Simu:DSInAWGN}, we have proved that our proposed PAWAS also minimizes average transmit power when serving delay-sensitive traffic. Since hybrid traffic is the simple aggregation of delay-sensitive and -insensitive traffic, we can conclude that our proposed PAWAS can minimize average transmit power efficiently for arbitrary traffic demand $(\lambda_{\rm{i}},\lambda_{\rm{s}},\tau_{\rm{max}})$.

Also, Fig. \ref{Fig:DelayInsAvePowPA} shows that in cases with low traffic demand, such as $\lambda_{\rm{i}}<700$ bit/s/Hz, SIMO with time-domain power allocation is suboptimal. This is because in low traffic demand cases, PAWAS select SIMO in most positions as train running along the railway, which will be shown in following sub-sections. In this cases, PAWAS degrades to SIMO with time-domain power allocation. However, by comparing Fig. \ref{Fig:DelayInsAvePowPA} and Fig. \ref{Fig:DelayInsAvePowEA}, we can observe that in cases with high traffic demand, such as $\lambda_{\rm{i}}>1000$ bit/s/Hz, MIMO with even power allocation is suboptimal. This is because in high traffic demand cases, PAWAS select MIMO in most positions and the water-filling coefficient $\eta$ is high, resulting in less effects of specific channel states compared to large $\eta$. This will also be shown in following sub-sections.

\subsection{Hybrid Traffic in  Sparse Scattering Scenarios}\label{Simu:HybridInAWGN}

\begin{figure}[tbp]
\centering
\subfigure[$(400,100,10)$]{ \label{Fig:Hyb400_100Pow}
\includegraphics[width=0.23\columnwidth]{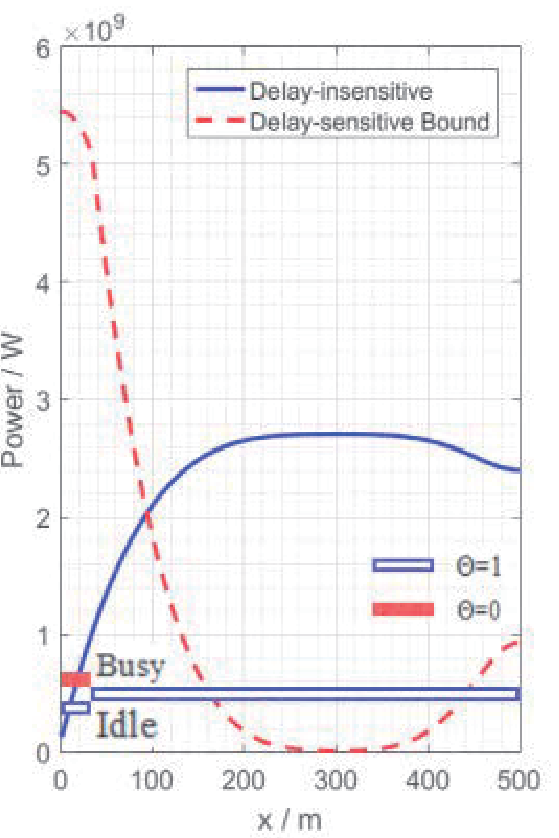}}
\subfigure[$(400,300,10)$]{ \label{Fig:Hyb400_300Pow}
\includegraphics[width=0.23\columnwidth]{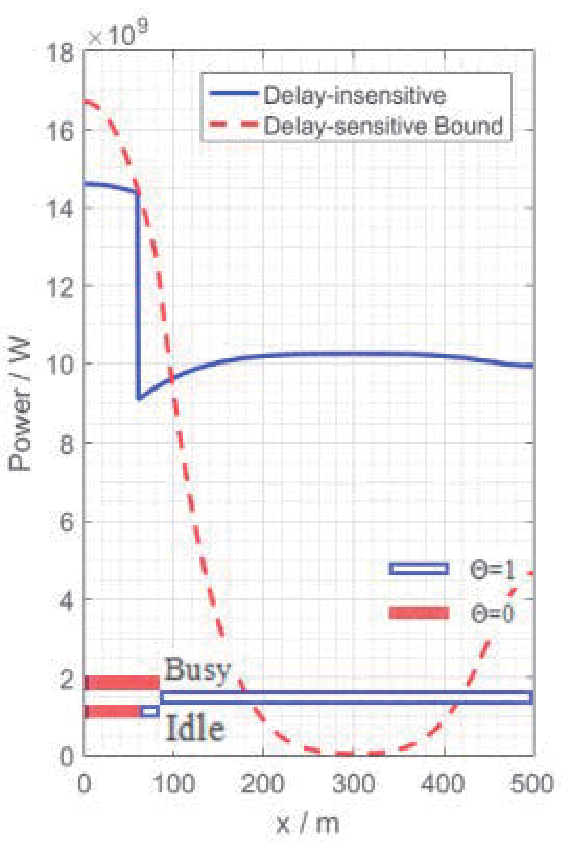}}
\subfigure[$(800,300,10)$]{ \label{Fig:Hyb800_300Pow}
\includegraphics[width=0.23\columnwidth]{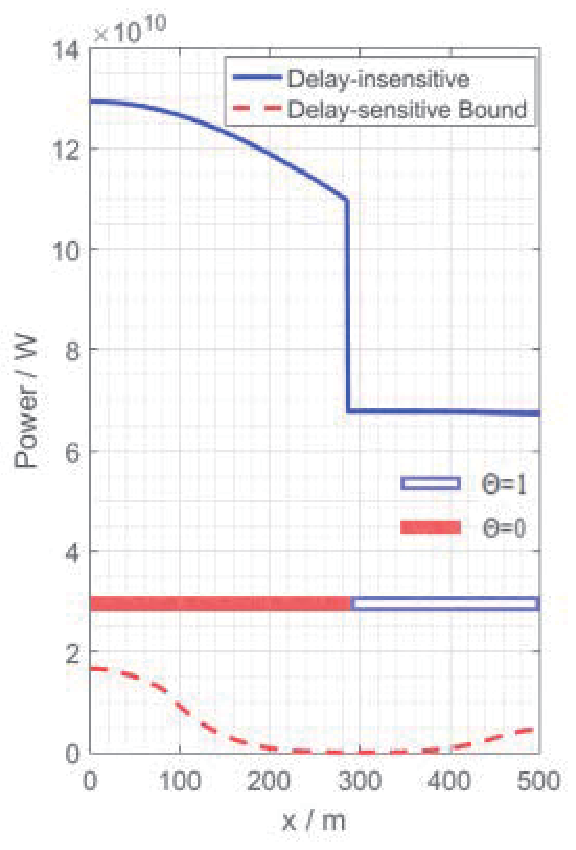}}
\subfigure[$(800,300,3)$]{ \label{Fig:Hyb800_300_3Pow}
\includegraphics[width=0.23\columnwidth]{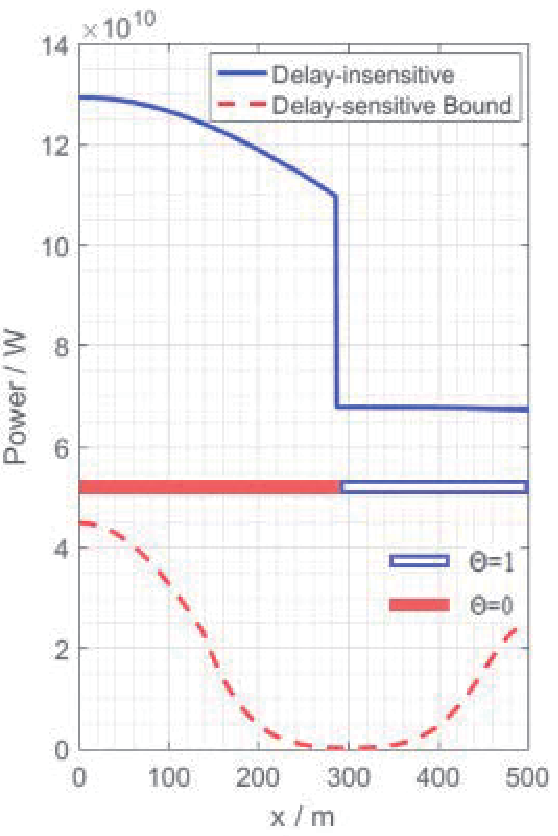}}
\caption{The optimal PAWAS for $(400,\lambda_{\rm{s}},10)$ are shown in (a) and (b), and the optimal PAWAS for $(800,300,\tau_{\rm{max}})$ are shown in (c) and (d). The colored bar at the bottom of each sub-figures demonstrates the selected antenna scheme.} \label{Fig:HybPowAndCapInVariousLambdaS}
\end{figure}

In Fig. \ref{Fig:HybPowAndCapInVariousLambdaS}, the performance of PAWAS for hybrid traffic is depicted. In each sub-figures, the arrival rates and delay constraints are labeled. The colored bars at the bottom of each sub-figures denote the selected antenna scheme. Specifically, the statuses 'Busy' and 'Idle' in Fig. \ref{Fig:HybPowAndCapInVariousLambdaS} denote the delay-sensitive traffic are being serviced or not. For instance, in Fig. \ref{Fig:Hyb400_100Pow}, when the system status is 'Idle', it select SIMO for $x\in[0,500]$, whereas when system status is 'Busy', MIMO is selected for $x\in[0,30]$. Similarly, in Fig. \ref{Fig:Hyb400_300Pow}, when system status is 'Idle', MIMO is selected for $x\in[0,60]$, whereas when system status is 'Busy', MIMO is selected for $x\in[0,80]$. Notice that for the concise of figures, we only depict the supposed allocated power relevant to stochastically arrived delay-sensitive traffic, which are denoted as delay-sensitive bounds and depicted in red-dash lines. It can be seen that the performance of delay-sensitive bounds agrees with the theoretical result shown in Section \ref{Simu:DSInAWGN}. Hence, in this part, we mainly discuss the power lines relevant to delay-insensitive traffic in 'Idle' status.

Observe Fig. \ref{Fig:Hyb400_100Pow}, SIMO is selected for all the positions in 'Idle' status. This is because the arrival rates of delay-sensitive and -insensitive traffic are limited, which leads to low capacity demand and $\hat{\eta}$ in \eqref{equ:WaterFillingCoefficientForHybrid}, $\hat{\mathcal{P}}_{\rm{I}} \leq \zeta_{\mathcal{P}}$ holds for $\forall x \in(0,500]$. However, in Fig. \ref{Fig:Hyb400_300Pow}, the arrival rate of delay-sensitive traffic is high, resulting in higher capacity demand and $\hat{\eta}$ in \eqref{equ:WaterFillingCoefficientForHybrid}. Therefore, when train is near $x=0$, $\hat{\mathcal{P}}_{\rm{I}}$ exceeds $\zeta_{\mathcal{P}}$ and MIMO is selected in 'Idle' status. Similarly, in Fig. \ref{Fig:Hyb800_300Pow} and Fig. \ref{Fig:Hyb800_300_3Pow}, due to the high capacity demands, system selects MIMO in 'Idle' status for $x\in[0,290]$. In summary, consistent with the performance of delay-sensitive traffic shown in Fig. \ref{Fig:DSFrom800to1500}, MIMO is selected in more positions for hybrid traffic with higher traffic demand, especially when train is near $x=0$.

For the power allocation shown in Fig. \ref{Fig:HybPowAndCapInVariousLambdaS}, it can be seen that when SIMO is selected, the power allocation is traditional waterfilling, whereas when MIMO is selected, the power allocation can be viewed as channel-inversion. Waterfilling for SIMO case is quite comprehensive. As illustrated in Section \ref{Simu:DSInAWGN}, when $x=300$, the channel gain achieves its maximum and allocating higher power in this case may enhance the energy efficiency. In high traffic demand cases shown in Fig. \ref{Fig:Hyb800_300Pow} and Fig. \ref{Fig:Hyb800_300_3Pow}, waterfilling degrades to even power allocation. This is because the inverse of channel gain ${1}/{\alpha_2}$ is less significant compared with $\hat{\eta}$. However, the optimal power allocation scheme in MIMO case is channel-inversion. This is unintuitive but reasonable. Before demonstrating that, we give following lemma first.

\begin{lem}\label{lem:MIMOPowerSpliting}
Define $f_1(x)=\log(a_1x^2+b_1x+1)$, $f_2(x)=\log(a_2x^2+b_2x+1)$ and $\Delta \rightarrow 0$. Let $b_1/a_1 < b_2/a_2$ and $x>2$. Then,
\begin{equation*}
f_1(x+\Delta)+f_2(x-\Delta) > f_1(x)+f_2(x).
\end{equation*}
\end{lem}
\begin{proof}
The proof is shown in Appendix \ref{Appen:MIMOPowerSpliting}.
\end{proof}
Then, we prove that channel-inversion in MIMO case is optimal by starting from evenly power allocation method. In \eqref{equ:MIMOCapacityLOS}, the proportion from ${(\alpha_1\alpha_2 - \beta^2)}/{4}$ to ${(\alpha_1 + \alpha_2)}/{2}$ decreases with increasing $|x-300|$, which can be verified by simulations. Consider the allocated power $\mathcal{P}_1$ and $\mathcal{P}_2$ in two positions $x_1$ and $x_2$, where $|x_1-300|>|x_2-300|$. We compare the provided capacities between evenly power allocation and channel-inversion. Let $\mathcal{P}_1=\mathcal{P}_2> 2~\rm{W}$, $\hat{\mathcal{P}}_1=\mathcal{P}_1+\Delta$, and $\hat{\mathcal{P}}_2=\mathcal{P}_2-\Delta$. According to Lemma \ref{lem:MIMOPowerSpliting}, $\hat{\mathcal{P}}_2$ and $\hat{\mathcal{P}}_2$ (i.e. channel-inversion) can provide more capacity than $\mathcal{P}_2$ and $\mathcal{P}_2$ (i.e. evenly power allocation), while consuming same transmit power. That is, the optimal power allocation strategy in MIMO case is channel-inversion.


Fig. \ref{Fig:Hyb800_300Pow} and Fig. \ref{Fig:Hyb800_300_3Pow} depict the special case shown in Corollary \ref{cor:TaumNotCare}, where $C(t) \geq \bar{L}\mu_{\rm{s}}$ holds for all $t\in[0,T/2]$. It can be observed that $\hat{\mathcal{P}}_{I}(t)$ and $C(t)$ are same between $\mathcal{T}(800,300,10)$ and $\mathcal{T}(800,300,3)$. That is, for cases with same $(\lambda_{\rm{s}}+\lambda_{\rm{i}})$ and sufficiently large $\lambda_{\rm{i}}$ compared to $\lambda_{\rm{s}}$, the PAWAS remains unchanged for arbitrary $\tau_{\rm{max}}$.

\subsection{Hybrid Traffic in Rich Scattering Scenarios}\label{Simu:HybridInNakagami}

\begin{figure}[tbp]
\centering
\subfigure[$m$=0.5]{ \label{Fig:NakaM0.5}
\includegraphics[width=0.23\columnwidth]{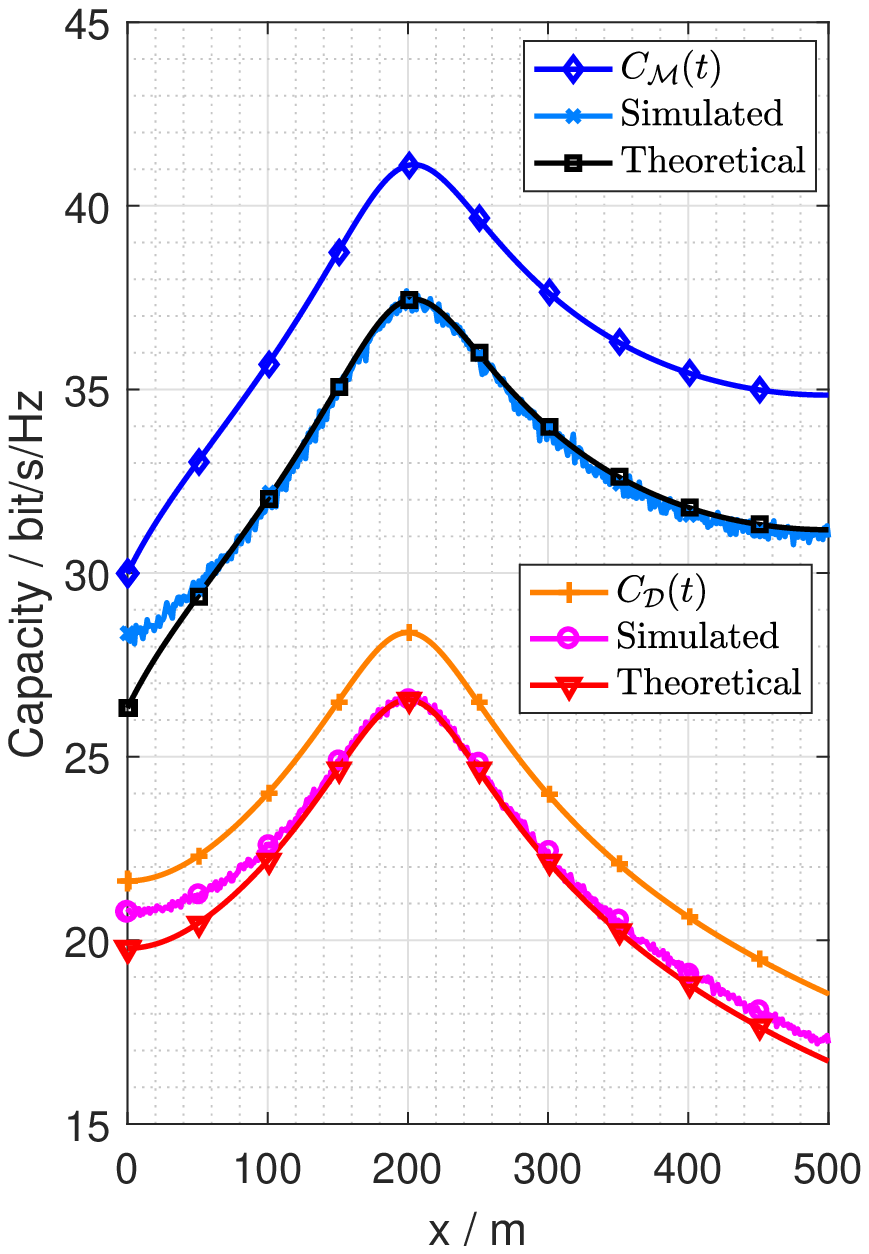}}
\subfigure[Err and Dist] { \label{Fig:CapacityErrOnM}
\includegraphics[width=0.23\columnwidth]{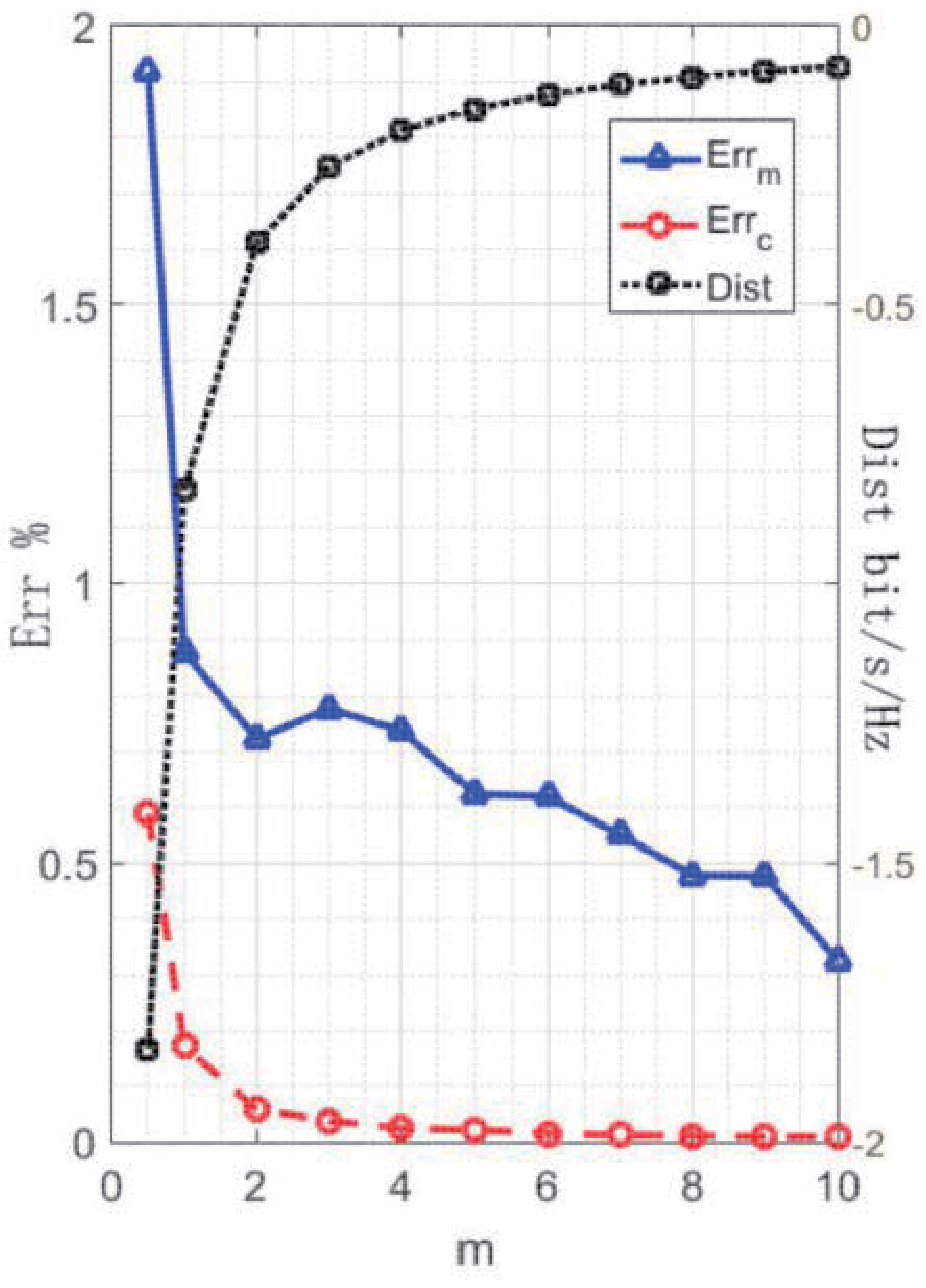}}
\subfigure[$(800,300,10)$ with $m$=0.5] { \label{Fig:Naka800_300_05}
\includegraphics[width=0.23\columnwidth]{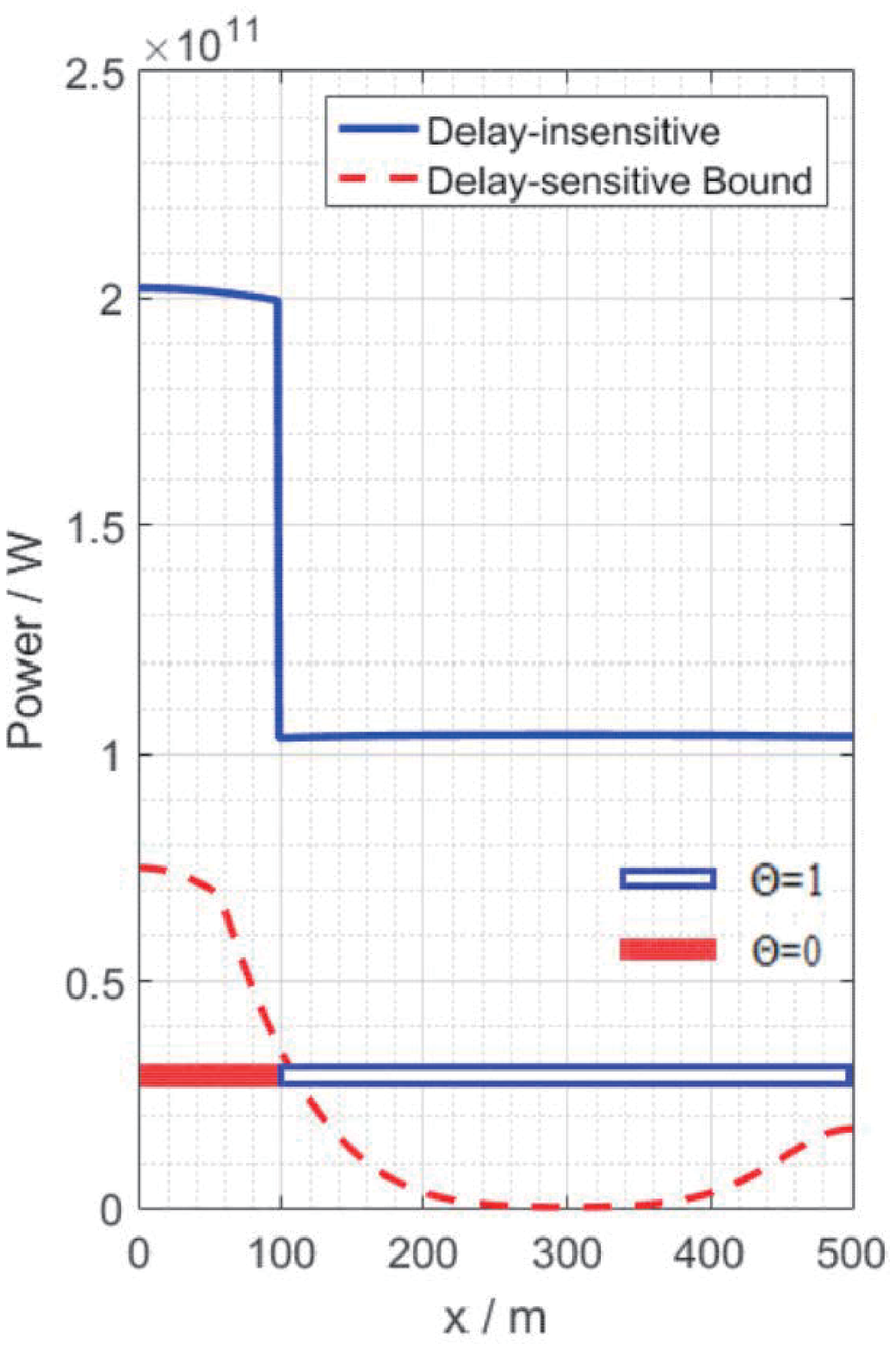}}
\subfigure[$(800,300,10)$ with $m$=3.5]{ \label{Fig:Naka800_300_35}
\includegraphics[width=0.23\columnwidth]{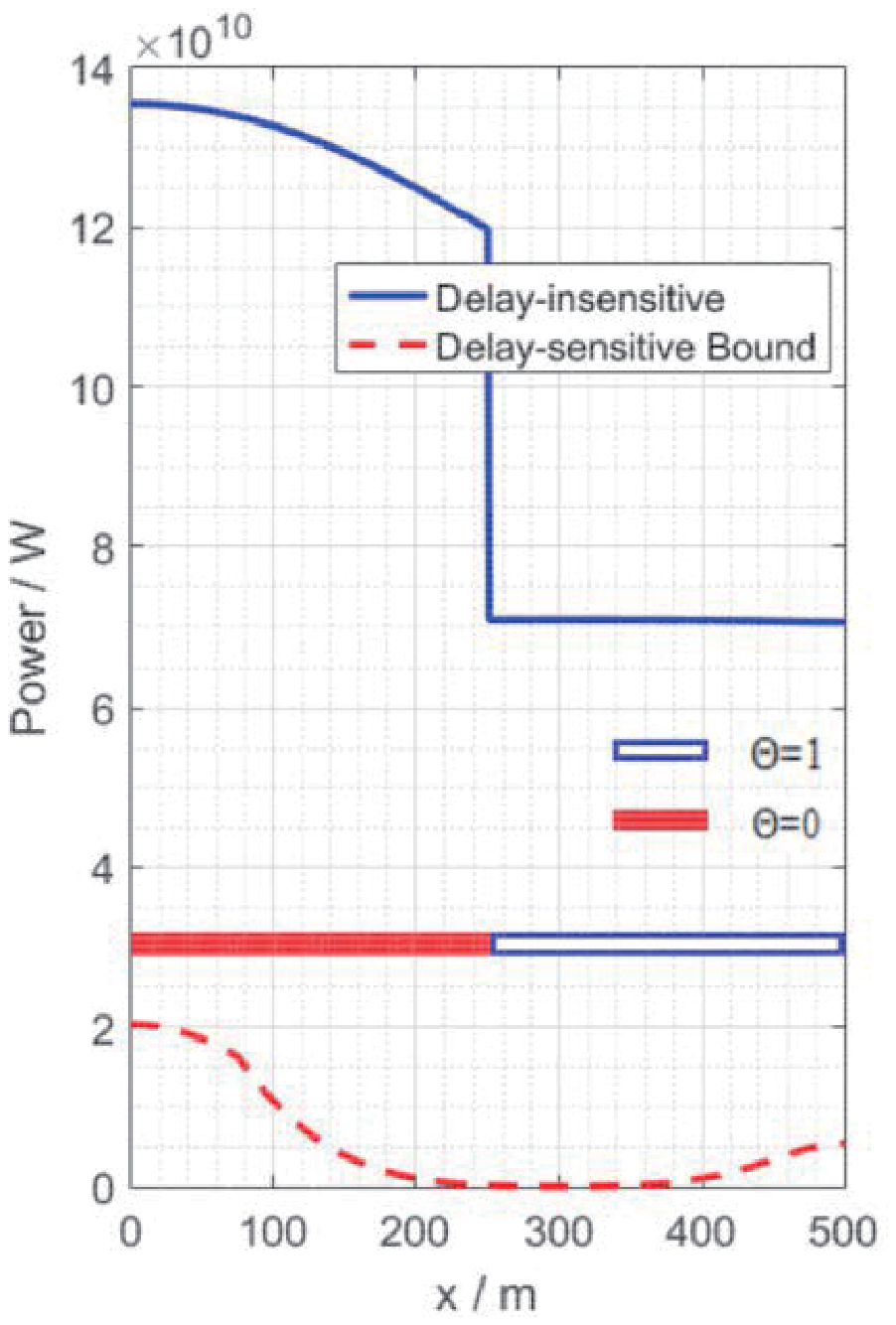}}
\caption{(a) compares the simulated and theoretical $\mathbb{C}_{1}(m,t)$ and $\mathbb{C}_{0}(m,t)$ by Monte-Carlo method, where $m$=0.5. Also, the capacities $C_1(t)$ and $C_0(t)$ in sparse scattering scenarios are depicted in (a). (b) shows the maximum relative errors and cumulative errors between simulation and theory in various $m$ by Monte-Carlo method. Also, the capacity distance is depicted. (c) compares the simulated and theoretical $\mathbb{C}_{1}(m,t)$ and $\mathbb{C}_{0}(m,t)$ by Monte-Carlo method, and the capacities $C_1(t)$ and $C_0(t)$ in AWGN scenarios are depicted in (c). (d) shows the maximum relative errors and the cumulative errors in various $m$ by Monte-Carlo method.} \label{Fig:NakaApproxInM}
\end{figure}
Firstly, we verify the validity of Theorem \ref{thm:HighSNRNakagamiCapacity} by Monte-Carlo simulations. To formulate the high SNR regime, we constrain the transmit power as 150 dBW. Let $m$=0.5 and the repetition number be 2000, Fig. \ref{Fig:NakaM0.5} depicts the simulated and theoretical $\mathbb{C}_{\rm{S}}(m,t)$ and $\mathbb{C}_{\rm{M}}(m,t)$ with Nakagami-$m$ fading. Also, the capacities of SIMO and MIMO in sparse scattering scenarios are also depicted. It can be seen that our theoretical results are quite explicit, except for $x\in[0,100]$. In these positions, due to the long distance between antenna pairs, the receive SNR is not high enough to ensure the accuracy of approximations in Theorem \ref{thm:HighSNRNakagamiCapacity}. In addition, define the relative maximum and cumulative errors as \eqref{equ:Errs}, which is shown at the top of next page. The superscripts s and t denote the simulated and theoretical capacities, respectively. Fig. \ref{Fig:CapacityErrOnM} shows the maximum relative errors and the cumulative errors in various $m$.  It can be seen that as $m$ increases, the approximation accuracy increases as well. Also, we define $\left(\psi(m)-\ln(m)\right)/\ln2$ as the capacity distance in sparse and rich scattering scenarios, which has been depicted in Fig. \ref{Fig:CapacityErrOnM} and decreases with respect to increasing $m$. We conclude that:

\begin{enumerate}
\item In severe fading cases, such as $m$=0.5, the approximation accuracy will become worse as SNR is relatively low. When SNR becomes high, the theoretical result gives satisfactory approximation to its real value. In flat fading cases, such as $m>2$, the approximation becomes more accurate compared to severe fading cases with the same SNR. That is, Theorem \ref{thm:HighSNRNakagamiCapacity} works well in high SNR or less severe fading scenarios.
\item As $m$ increases, the capacities in Nakagami-$m$ case approach to those in sparse scattering scenarios. Specifically, when $m\rightarrow \infty$, capacities in rich scattering scenarios is the same as those in sparse scattering scenarios, and hence the PAWAS is consistent in both scenarios.
\end{enumerate}

\begin{figure*}[!t]
\normalsize
\begin{equation}\label{equ:Errs}
\begin{split}
&\rm{Err}_m = \max_{t\in[0,T/2]} \max\left(\frac{|\mathbb{C}^{s}_{\rm{S}}(m,t)-\mathbb{C}^{t}_{\rm{S}}(m,t)|}{\mathbb{C}^{s}_{\rm{S}}(m,t)},\frac{|\mathbb{C}^{s}_{\rm{M}}(m,t)-\mathbb{C}^{t}_{\rm{M}}(m,t)|}{\mathbb{C}^{s}_{\rm{M}}(m,t)}\right),\\
&\rm{Err}_c = \max\left(\frac{|\int_{0}^{T/2}(\mathbb{C}^{s}_{\rm{S}}(m,t)-\mathbb{C}^{t}_{\rm{S}}(m,t))dt|}{\int_{0}^{T/2}\mathbb{C}^{s}_{\rm{S}}(m,t)dt},\frac{|\int_{0}^{T/2}(\mathbb{C}^{s}_{\rm{M}}(m,t)-\mathbb{C}^{t}_{\rm{M}}(m,t))dt|}{\int_{0}^{T/2}\mathbb{C}^{s}_{\rm{M}}(m,t)dt}\right).\\
\end{split}
\end{equation}
\hrulefill
\vspace*{4pt}
\end{figure*}

As shown in Fig. \ref{Fig:Naka800_300_05} and Fig. \ref{Fig:Naka800_300_35}, the optimal PAWAS for $m$=0.5 and 3.5 are depicted. According to Fig. \ref{Fig:CapacityErrOnM}, when $m$=0.5, small-scale fading leads to greatly capacity losses in SIMO and MIMO by 1.8 bit/s/Hz and 3.6 bit/s/Hz, respectively. Hence, compared to the simulation results shown in Fig. \ref{Fig:Hyb800_300Pow}, the allocated powers for delay-sensitive and -insensitive traffic in Fig. \ref{Fig:Naka800_300_05} is sharply increased. In addition, since the capacity loss in $2\times2$ MIMO doubles that in $1\times2$ SIMO, the switching points between SIMO and MIMO in rich scattering scenarios are closer to $x=0$ than those in sparse fading scenarios. In summary, we conclude that server fading increases the allocated power in both SIMO and MIMO and leads system to select SIMO in more positions.

\subsection{Effects of Power Constraints}\label{Sec:PowerConstraints}

\begin{figure}[tbp]
\centering
\subfigure[108 dB] { \label{Fig:pm108}
\includegraphics[width=0.23\columnwidth]{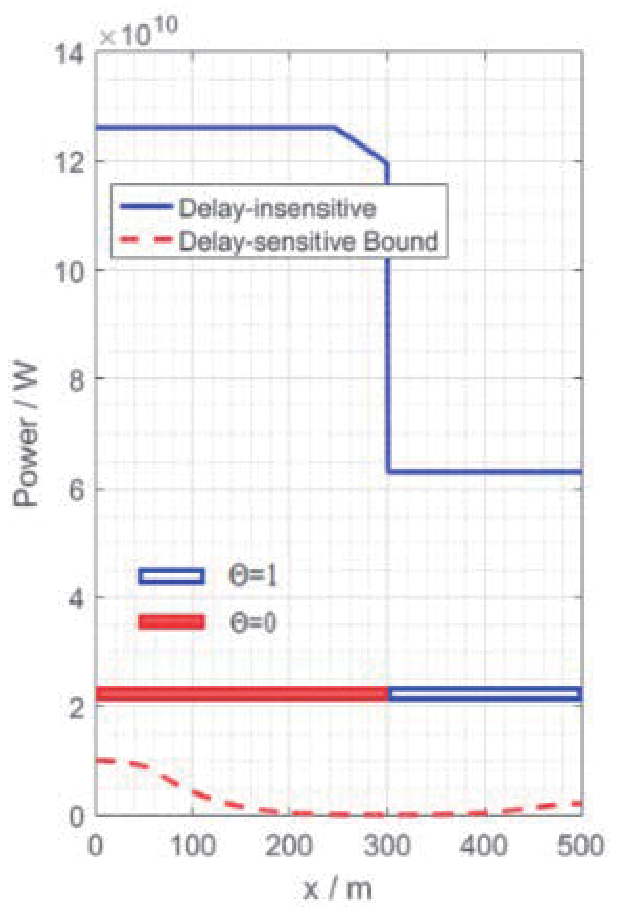}}
\subfigure[107.7 dB]{ \label{Fig:pm1077}
\includegraphics[width=0.23\columnwidth]{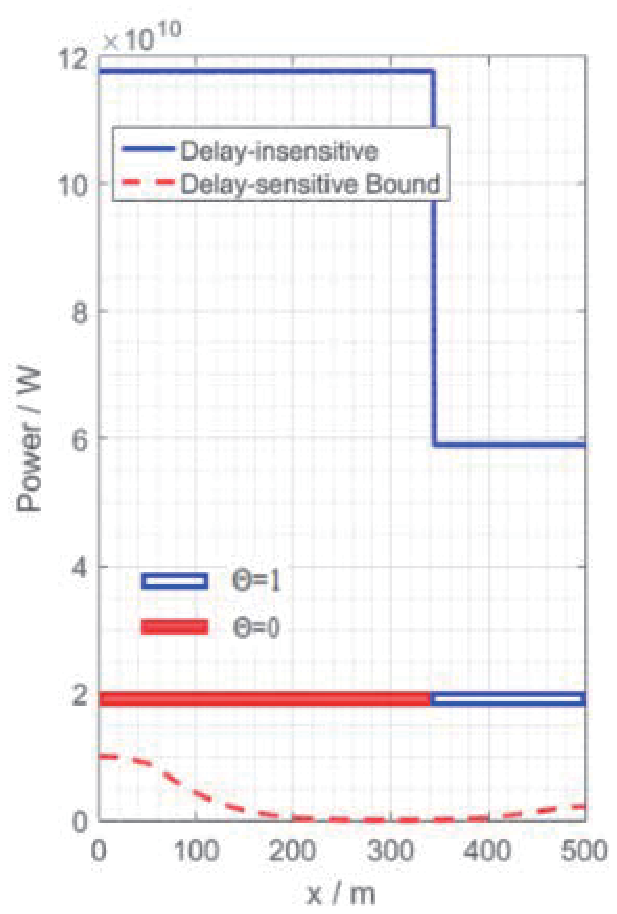}}
\subfigure[107.3 dB] { \label{Fig:pm1073}
\includegraphics[width=0.23\columnwidth]{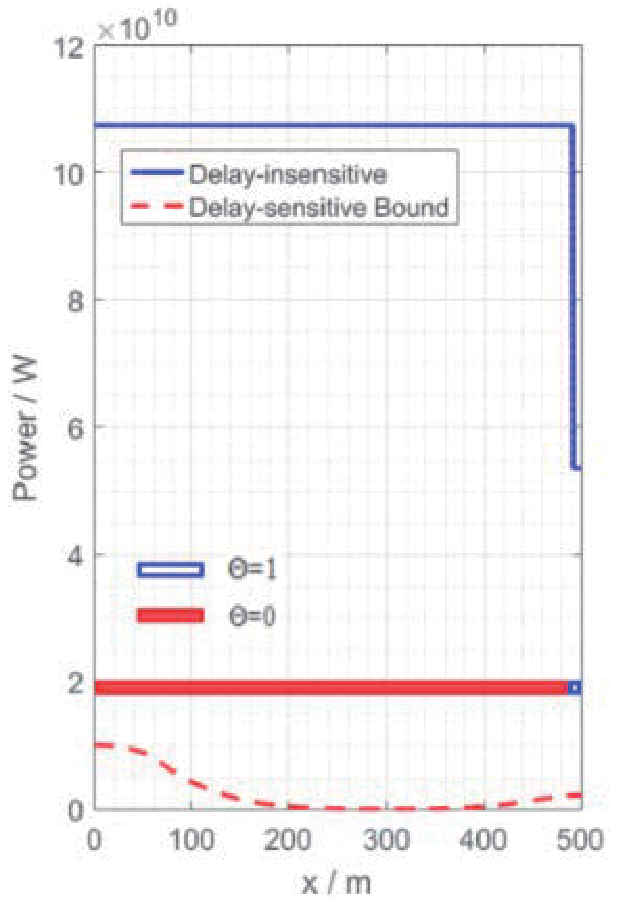}}
\subfigure[]{ \label{Fig:PmaxPosition}
\includegraphics[width=0.23\columnwidth]{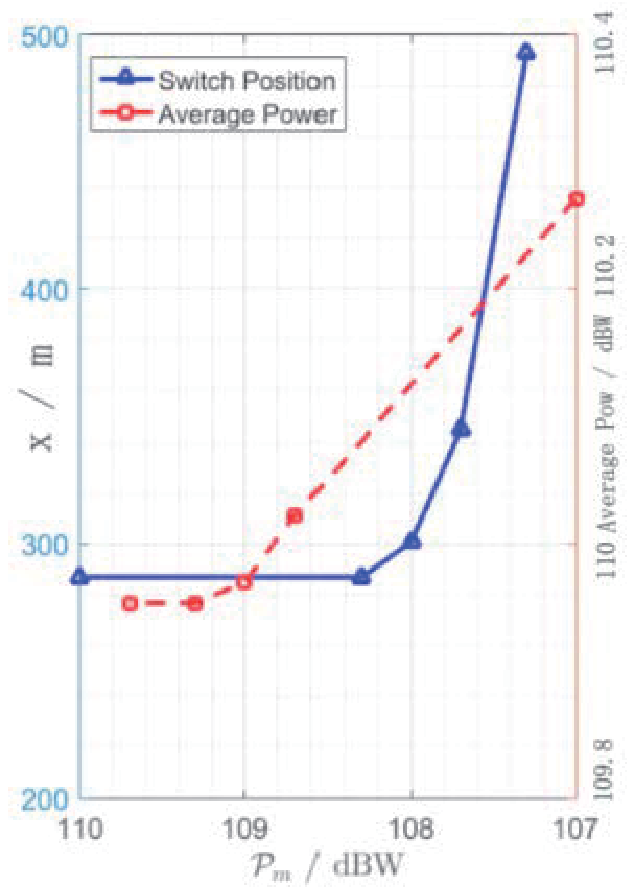}}
\caption{Sub-figures (a)-(c) depict the allocated power with $\mathcal{P}_m=$108.3 dB, , 108 dB and 107.3 dB. (d) shows the switching point positions and average transmit power with varying $\mathcal{P}_m$. In simulations, the traffic demand is $(800,300,10)$.} \label{Fig:PmaxEffects}
\end{figure}

In this part, we evaluate the effects of $\mathcal{P}_m$ on our proposed PAWAS. Since we assume that $\mathcal{P}_m$ is sufficiently large to provide stable service to delay-sensitive traffic, we take the hybrid traffic case with $C(t)\geq\bar{L}\big( \lambda_{\rm{s}} + {1}/{\tau_{\rm{max}}} \big)$ in sparse scattering scenarios for example. The optimal PAWAS with $\mathcal{P}_m=$108.3 dB, 108 dB and 107.3 dB and the corresponding switching point positions and average transmit power are shown in Fig. \ref{Fig:PmaxEffects}. From Fig. \ref{Fig:pm108}-Fig. \ref{Fig:pm1073}, we can observe that with lower $\mathcal{P}_m$, the maximum allocated power decreases and system is forced to select MIMO in more positions. Fig. \ref{Fig:PmaxPosition} shows the switching point positions corresponding to various $\mathcal{P}_m$. It can be seen that the positions are quite sensitive to $\mathcal{P}_m$. That is, when $\mathcal{P}_m=108.3$ dB, the optimal PAWAS is the same as that without power constraint. However, when $\mathcal{P}_m=107.3$ dB, the optimal PAWAS selects MIMO and allocates power evenly by $2\mathcal{P}_m$ in almost all positions. When $\mathcal{P}_m<107.3$ dB, there won't be any feasible power allocation strategy satisfying traffic demand $(800,300,10)$. Hence, we suggest that the maximum power constraints is supposed to be sufficiently large to meet the traffic demands. Otherwise, the PAWAS will be sharply degraded to even time-domain power allocation with MIMO, where $\mathcal{P}(t)=2\mathcal{P}_m$ and $t\in[0,T/2]$. Also, it can be seen that the average transmit power increases slightly when $\mathcal{P}_m$ varies from 110 dB to 107.3 dB, meaning that in this high traffic demand case, selecting MIMO and allocating power evenly in all positions is sub-optimal. This agrees with the theoretical results shown in sub-section \ref{Sec:PAWASInDSAWGN}.

\section{Conclusion}\label{Sec:Conclusion}

This paper focused on high-speed railway communication system with DAS and proposed a joint optimal power allocation with antenna selection (PAWAS) method to minimize average transmit power, where three different traffic patterns, i.e. delay-sensitive and -insensitive traffic, and hybrid traffic were considered. By defining \emph{effective channel gain}, we firstly derived the PAWAS for delay-sensitive and -insensitive traffic in sparse scattering scenarios, respectively. Then, we prove that the PAWAS for hybrid traffic is simple aggregation of that for its delay-sensitive and -insensitive part. We also proved that the PAWAS for delay-sensitive traffic can be viewed as the generalization of channel-inversion but with transmit antenna selection. For delay-insensitive traffic, when MIMO is selected, the power allocation can be treated as channel-inversion, whereas if SIMO is selected, it becomes traditional water-filling. In rich scattering scenarios, we modeled the small-scale fading channel with Nakagami-$m$ and showed its effects to the loss of ergodic capacity. Simulation results showed if $m\rightarrow\infty$, the theoretical performance in rich scattering scenarios are consistent with that in sparse scattering scenarios.

By comparing with other power allocation methods, we showed that our proposed method can provide lower average transmit power in arbitrary traffic demand cases. We also demonstrate that: 1) in low traffic demand cases, SIMO with time-domain power allocation is sub-optimal; 2) in high traffic demand cases, MIMO with even power allocation is also sub-optimal. It should be noted that in both sparse and rich scattering scenarios, the maximum transmit power constraint caused by limited module power may greatly affect the results of PAWAS and is supposed to be sufficiently large to meet the traffic demands.



\begin{appendices}

\section{Proof for Proposition \ref{prop:NonDelaySensitiveOptimalPowerAndMDS}}\label{Appen:PowerAllocationForNDS}
The Lagrangian function of $\mathbf{P}_{1\text{-}\rm{A}}$ is
\begin{equation}\label{equ:LagrangianFunction}
\begin{split}
L(\mathcal{P}(t), \eta \ln2) &= \int_{0}^{T/2} \Big\{ \mathcal{P}(t) - \eta \ln2 \Big( \log_2\big(1+\Gamma\big(\mathcal{P}(t),t\big)\mathcal{P}(t)\big) + \lambda_{\rm{i}} \Big) \Big\}dt\\
&\doteq \int_{0}^{T/2} \frac{\partial L(\mathcal{P}(t), \eta \ln2)}{\partial t}dt,\\
\end{split}
\end{equation}
where $\eta \ln2$ is the Lagrangian multiplier, and $\eta$ is denoted as the waterfilling coefficient. Minimizing this Lagrangian function is equivalent to minimizing the power allocation and antenna selection scheme point-wise. Differentiating of ${\partial L(\mathcal{P}(t), \eta \ln2)}/{\partial t}$ with respect to $\mathcal{P}(t)$ and let it be zero, we have
\begin{equation*}
\begin{split}
\frac{{\partial L(\mathcal{P}(t), \eta \ln2)}/{\partial t}}{\partial \mathcal{P}(t)} &= 1 - \eta \frac{\Gamma(\mathcal{P}(t),t) + \mathcal{P}(t){\partial \Gamma(\mathcal{P}(t),t)}/{\partial \mathcal{P}(t)}}{1+\Gamma(\mathcal{P}(t),t)\mathcal{P}(t)} = 0\\
\Rightarrow \mathcal{P}(t) &= \frac{\eta \Gamma(\mathcal{P}(t),t) - 1}{\Gamma(\mathcal{P}(t),t) - \eta \frac{\partial \Gamma(\mathcal{P}(t),t)}{\partial \mathcal{P}(t)}}.
\end{split}
\end{equation*}
Since $\mathcal{P}(t)$ is non-negative, the optimal waterfilling power is \eqref{equ:OptimalWaterFillingForNDS}. Considering the constraints $0 \leq \mathcal{P}(t) \leq \mathcal{P}_{\rm{max}}$, the instantaneous transmit power can be expressed as \eqref{equ:NDSoptPower}. Substitute \eqref{equ:NDSoptPower} into $\mathbf{P}_{1\text{-}\rm{A}}$, $\eta$ can be determined. This completes the proof.

\section{Proof for Lemma \ref{lem:MaxDelayEfficiency}}\label{Appen:MaxDelayEfficient}
Minimizing the average transmit energy in $\mathbf{P}_{1\text{-}B}$ is equivalent to minimizing the transmit power point-wise. That is, $p_{\rm{on}} \mathcal{P}(t)$ needs to be minimized. According to \cite{heyman1968optimal}, $p_{\rm{on}} = \rho$. Then, the optimization problem can be expressed as
\begin{equation}\label{equ:ConstentCapacityOptimization}
\min ~ \rho \mathcal{P}(t), \quad s.t.~\tau \leq \tau_{\rm{max}}.
\end{equation}
Substitute $C_{\rm{S}}(t)=C_{\rm{M}}(t)=\bar{L}\mu_{\rm{s}}$, \eqref{equ:MIMOCapacityLOS} and \eqref{equ:SIMOCapacityLOS} into \eqref{equ:ConstentCapacityOptimization}, $\rho \mathcal{P}(t)$ is given by
\begin{equation*}
\begin{split}
&\rho \mathcal{P}(t) =\bigg\{
\begin{array}{lc}
{\mathcal{P}(t)}/{\log_2 \big\{ \frac{\alpha_1\alpha_2 - \beta^2}{4} \mathcal{P}(t)^2 + \frac{\alpha_1 + \alpha_2}{2} \mathcal{P}(t) + 1 \big\}},    & 0\\
{\mathcal{P}(t)}/{\log_2\big( 1+\mathcal{P}(t)\alpha_2 \big)}  ,  & 1.\\
\end{array}
\end{split}
\end{equation*}
For SIMO,
\begin{equation*}
\begin{split}
\frac{\partial \left(\rho \mathcal{P}(t)\right)}{\partial \mathcal{P}(t)}&=\frac{\log_2(1+\alpha_2\mathcal{P}(t))-{\alpha_2\mathcal{P}(t)}/{\ln2(1+\alpha_2\mathcal{P}(t))}}{\log_2\big( 1+\mathcal{P}(t)\alpha_2 \big)^2}\\
&\overset{(a)}{\thickapprox} \frac{\log_2(1+\alpha_2\mathcal{P}(t))-{1}/{\ln2}}{\log_2\big( 1+\mathcal{P}(t)\alpha_2 \big)^2}~ \overset{(b)}{>}0.\\
\end{split}
\end{equation*}
The approximation (a) and inequality (b) comes from the high transmit power assumption. Similarly, ${\partial \left(\rho \mathcal{P}(t)\right)}/{\partial \mathcal{P}(t)} > 0$ can be also proved for MIMO under the high transmit power assumption.

Since $\tau$ in \eqref{equ:QueDelay} decreases with increasing $\mu$, the minimum $\rho \mathcal{P}(t)$ subjecting to $\tau \leq \tau_{\rm{max}}$ is achieved when $\tau = \tau_{\rm{max}}$. This completes the proof.

\section{Proof for Theorem \ref{thm:HighSNRNakagamiCapacity}}\label{Appen:NakagamiCapacity}
The ergodic capacity of distributed MIMO system, constituted by several centralized located receive antennas and several radio ports located far apart, is shown in \cite{zhong2009capacity}. However, in our considered HSR communication system, the receive antennas (i.e. the MRs) are distributed as well, and hence the ergodic capacity of proposed distributed antenna systems needs to be reconsidered.

To solve the $\log_2\det(\cdot)$ function in \eqref{equ:MIMOCapacityExpressionOriginal}, we give the following lemma first.
\begin{lem}\label{lem:diagonalSimplifyLem}
Let $\mathbf{Q}$ be an $n \times n$ Hermitian matrix with diagonal elements denoted by $d_{i,i}$ ($i=1,2,\dots,n$) and the eigenvalues denoted by $\lambda_{i}$ ($i=1,2,\dots,n$). Define the diagonal matrix $\mathbf{\Lambda}$ is formed by $\lambda_{i}$. That is, the diagonal elements $\Lambda_{i,i} = \lambda_{i}$ ($i=1,2,\dots,n$). Then,
\begin{equation}
\sum_{i=1}^{n} \log_2\left( 1 + d_{i,i} \right) \geq \sum_{i=1}^{n} \log_2\left( 1 + \lambda_{i} \right).
\end{equation}
\end{lem}
\begin{proof}
Since the function $g(x) = \log_2(1+ax)$ is concave for $a>0$, substitute \cite[Lemma 2]{zhong2009capacity}  into \cite[Definition 3]{zhong2009capacity}, Lemma \ref{lem:diagonalSimplifyLem} can be easily proved.
\end{proof}

The channel fading matrix with small-scale fading can be rewritten as
\begin{equation}\label{equ:NakagamiMIMOFading}
\mathbf{H} = \begin{bmatrix} h_{11}n_{11} & h_{12}n_{12} \\ h_{21}n_{21} & h_{22}n_{22} \end{bmatrix}.
\end{equation}
Note that the time index is omitted, and $n_{ij}$ ($i,j=1,2$) obeys the Nakagami-m distribution $\gamma(m,\Omega)$. Because the small-scale fading is normalized, $\Omega=1$. Then, substitute \eqref{equ:NakagamiMIMOFading} and $\mathcal{P}_1 = \mathcal{P}_2 = \mathcal{P}/2$ into \eqref{equ:MIMOCapacityExpressionOriginal}, the capacity of MIMO can be written as
\begin{equation*}
\begin{split}
&\mathbb{E}\Big(C\big( \mathcal{P}, m, t \big)\Big) \\
&= \mathbb{E}\bigg(\log_2\Big( \mathbf{I} + \begin{bmatrix} \mathcal{P}/2 & 0 \\ 0 & \mathcal{P}/2 \end{bmatrix} \mathbf{H}^{\dagger}\mathbf{H} \Big)\bigg)\\
&\overset{(a)}{\approx} \mathbb{E}\bigg(\log_2\Big( \begin{bmatrix} \mathcal{P}/2 & 0 \\ 0 & \mathcal{P}/2 \end{bmatrix} \mathbf{H}^{\dagger}\mathbf{H} \Big)\bigg)\\
&\overset{(b)}{\leq} \log_2( \frac{\mathcal{P}\alpha_1}{2}) + \log_2( \frac{\mathcal{P}\alpha_2}{2}) + \mathbb{E}\Big(\log_2\big(\frac{h_{11}^2n_{11}^2 + h_{21}^2n_{21}^2}{h_{11}^2+h_{21}^2}\big)\Big)+ \mathbb{E}\Big(\log_2\big(\frac{h_{12}^2n_{12}^2 + h_{22}^2n_{22}^2}{h_{12}^2+h_{22}^2}\big)\Big)\\
&\overset{(c)}{=} C_{\rm{M}}(t) + {2\big(\psi(m)-\ln(m)\big)}/{\ln2}.
\end{split}
\end{equation*}
Note that the time index of $\mathcal{P}(t)$ is omitted for the concise of equation. The approximation (a) comes from the high SNR assumption. The inequality (b) follows Lemma \ref{lem:diagonalSimplifyLem}. In high SNR regime, $C_{\rm{M}}(t) = \log_2( {\mathcal{P}\alpha_1}/{2}) + \log_2( {\mathcal{P}\alpha_2}/{2})$. Since $h_{ij} \sim \gamma(m,1)$, it's clear that ${(h_{11}^2n_{11}^2 + h_{21}^2n_{21}^2)}/{(h_{11}^2+h_{21}^2)}\sim \gamma(m,1)$. According to \cite{zhong2009capacity}, the expectation is equal to $\big(\psi(m)-\ln(m)\big)/\ln2$, where $\psi(m)$ is the digamma function. Then, equation (c) can be easily derived.

Similarly, substitute $h_{12},~h_{22}\sim\gamma(m,1)$ and $\mathcal{P}_2 = \mathcal{P}$ into \eqref{equ:SIMOCapacityLOS}, \eqref{equ:MRCCapacityWithSmallScaleFaing} can be derived. This completes the proof.

\section{Proof for Lemma \ref{lem:MIMOPowerSpliting}}\label{Appen:MIMOPowerSpliting}
The left hand of inequation can be written as
\begin{equation*}
\begin{split}
&f_1(x+\Delta)+f_2(x-\Delta)\\
&\overset{(a)}{=}\log(a_1x^2+b_1x+1 + 2(a_1\Delta+b_1\Delta)) + \log(a_2x^2+b_2x+1 - 2(a_2\Delta+b_2\Delta))\\
&\overset{(b)}{=} f_1(x)+f_2(x) + \Delta\Big(\frac{(2a_1+b_1)}{a_1x^2+b_1x+1} - \frac{(2a_2+b_2)}{a_2x^2+b_2x+1}\Big)\overset{(c)}{>}f_1(x)+f_2(x).
\end{split}
\end{equation*}
In (a), $\Delta^2$ is omitted for $\Delta\rightarrow 0$, and (b) comes from $\log(1+x)=x$ ($x\rightarrow0$). Substitute $a_1 < a_2$, $b_1 < b_2$ and $x>2$ into (b), (c) can be easily derived. This completes the proof.
\end{appendices}

\bibliography{bibfile}



\end{document}